%% file: main.tex
\documentclass{article}
\usepackage{kr}

\usepackage{times}
\usepackage{soul}
\usepackage{url}
\usepackage[hidelinks]{hyperref}
\usepackage[utf8]{inputenc}
\usepackage[small]{caption}
\usepackage{graphicx}
\usepackage{amsmath}
\usepackage{amsthm}
\usepackage{booktabs}
\usepackage{algorithm}
\usepackage{algorithmic}
\usepackage{etoolbox} % for \AtBeginEnvironment

\AtBeginEnvironment{align}{\footnotesize}
\AtBeginEnvironment{align*}{\footnotesize}

\allowdisplaybreaks

\usepackage{fancyhdr} % for custom headers/footers

\author{%
Nina Pardal$^{1,2}$\and
Santiago Cifuentes$^2$\and
Edwin Pin$^{3}$\and
Maria Vanina Martinez$^4$ \and
Sergio Abriola$^2$\\
\affiliations
$^1$University of Sheffield, UK\\
$^2$Instituto de Ciencias de la Computación, UBA-CONICET, Argentina\\
$^3$Departamento de Matemática, CONICET, Universidad de Buenos Aires\\
$^4$Artificial Intelligence Research Institute, IIIA-CSIC, Bellaterra, Spain \\
\emails
n.pardal@sheffield.ac.uk,
\{scifuentes, sabriola, epin\}@dc.uba.ar,
vmartinez@iiia.csic.es
}

\usepackage{amssymb}
\usepackage{amsfonts}
\usepackage{amsmath}
\usepackage{amsthm}
\usepackage{stmaryrd}
\usepackage{graphicx}
\usepackage[utf8]{inputenc}
\usepackage[dvipsnames]{xcolor}
\usepackage{xspace}
\usepackage{hyperref}
\hypersetup{
    colorlinks=true,
    linkcolor=blue,
    filecolor=magenta,      
    urlcolor=cyan,
    }
\usepackage{tikz}
\usetikzlibrary{trees}
\usetikzlibrary{fit,positioning}

\usepackage{pgfplots}
\pgfplotsset{compat = newest}

\input{macros}

\title{Computational Complexity of Preferred Subset Repairs on Data-Graphs}

\begin{document}

\maketitle

\begin{abstract}
Preferences are a pivotal component in practical reasoning, especially in tasks that involve decision-making over different options or courses of action that could be pursued. In this work, we focus on repairing and querying inconsistent knowledge bases in the form of graph databases, which involves finding a way to solve conflicts in the knowledge base and considering answers that are entailed from every possible repair, respectively. Without {\it a priori} domain knowledge, all possible repairs are equally preferred. Though that may be adequate for some settings, it seems reasonable to establish and exploit some form of preference order among the potential repairs. We study the problem of computing prioritized repairs over graph databases with data values, using a notion of consistency based on \Gregxpath expressions as integrity constraints. We present several preference criteria based on the standard subset repair semantics, incorporating weights, multisets, and set-based priority levels. We show that it is possible to maintain the same computational complexity as in the case where no preference criterion is available for exploitation. Finally, we explore the complexity of consistent query answering in this setting and obtain tight lower and upper bounds for all the preference criteria introduced.
% The problem of repairing inconsistent knowledge bases has a long history within the communities of database theory and knowledge representation and reasoning, especially from the perspective of structured data.
% However, as the data available in real-world domains becomes more complex and interconnected, the need naturally arises for developing new types of repositories, representation languages, and semantics, to allow for more suitable ways to query and reason about it. 
% Graph databases provide an effective way to represent relationships among semi-structured data, and allow processing and querying these connections efficiently.  
% In this work, we focus on the problem of computing prioritized repairs over graph databases with data values, using a notion of consistency based on \Gregxpath expressions as integrity constraints. We present several preference criteria based on the standard subset repair semantics, incorporating weights, multisets, and set-based priority levels. We study the most common repairing tasks, showing that it is possible to maintain the same computational complexity as in the case where no preference criterion is available for exploitation. To complete the picture, we explore the complexity of consistent query answering in this setting and obtain tight lower and upper bounds for all the preference criteria introduced.
\end{abstract}

\section{Introduction}

%\santi{Sacar demos, poner idea
%
% Agregar notaciones: $V_G$ para referirnos a los nodos de $G$ y la idea de grafo inducido por un conjunto de nodos.
% }

Many current real-world applications give rise to massive data repositories whose inherent nature leads to the presence of uncertainty or inconsistency. %, independent of the particular domain of application.
The high volume of interconnected data expected to be handled by intelligent management systems represents an important obstacle when attempting to reason with these massive repositories. Additionally, the richness of available data requires developing applications that go well beyond semantic indexing and search, and involves advanced reasoning tasks on top of existing data, some of which may need to exploit domain knowledge, for instance in the form of preferences, in order to be effective in real-world applications. 

Enriching the expressive power of traditional relational databases, alternative data models such as graph databases are deployed in many modern applications where the topology of the data is as important as the data itself, such as social networks analysis~\cite{fan2012graph}, data provenance~\cite{anand2010techniques}, and the Semantic Web ~\cite{arenas2011querying}. The structure of the database is commonly queried through navigational languages such as \textit{regular path queries} or RPQs~\cite{barcelo2013querying} that can capture pairs of nodes connected by some specific kind of path. These query languages can be extended to be more expressive, with a usual trade-off in the complexity of the evaluation. For example, C2RPQs  are a natural extension of RPQs that add to the language the capability of traversing edges backwards and closing the expressions under conjunction (similar to relational CQs). When considering the combined complexity (i.e. both the data-graph and the query are considered as inputs of the problem), the evaluation of RPQs can be done in polynomial time while evaluating C2RPQs is an \NP-complete problem~\cite{barcelo2012expressive}.  

As in the relational case, the data is expected to preserve some semantic structure related to the reality it intends to represent. These \textit{integrity constraints} can be expressed in graph databases via \textit{path constraints}~\cite{abiteboul1997regular,buneman2000path}.
RPQs and its most common extensions (C2RPQs and NREs \cite{barcelo2012relative}) can only act upon the edges of the graph, leaving behind any possible interaction with data values in the nodes. This inspired the design of query languages for \emph{data-graphs} (i.e. graph databases where data lies both in the edges and the nodes), such as REMs and \Gregxpath~\cite{libkin2016querying}.
When a database does not satisfy the imposed integrity constraints, a possible approach is to search for a database that satisfies the constraints and that is `as close as possible to the original one'. This new database is called a \emph{repair}~\cite{arenas1999consistent}, and to define it properly one has to precisely define the meaning of `as close as possible'. 
Several notions of repairs have been developed in the literature, among others, set-based repairs~\cite{Cate:2015}, attribute-based repairs~\cite{Wijsen:2003}, and cardinality-based repairs~\cite{Lopatenko07complexityof}.
When considering set-based repairs $G'$ of a graph database $G$ under a set of expressions $R$, a natural restriction of the problem is when $G'$ is a subgraph of $G$; such repairs are usually called \textit{subset} repairs~\cite{barcelo2017data,Cate:2015}.
Since repairs may, in general, not be unique, 
some applications may require a methodology to choose the more adequate(s) among them. One way to do this is to impose an ordering over the set of repairs and look for an optimum one over such ordering. There is a significant body of work on preferred repairs
for relational databases~\cite{flesca2007preferred,staworko2012prioritized,DBLP:journals/tods/Chomicki03} and other types of logic-based formalisms~\cite{brewka89,bienvenu2014querying,DBLP:conf/kr/LukasiewiczMM23}. 
However, to the best of our knowledge, there is no such work focused on graph databases or data-graphs.
We address the problem of finding a preferred subset repair with different types of prioritization, which encompasses some prominent preorders studied in the literature~\cite{DBLP:conf/kr/LukasiewiczMM23}. Moreover, we propose a preference criteria based on comparing the underlying multisets of facts of the data-graphs.

Furthermore, as repairing a database may not always be possible or desirable, e.g.\ we may not have permissions to modify the data, alternative approaches are needed to be able to query a potentially inconsistent database without modifying it. \emph{Consistent query answering} (or CQA),  puts in practice the notion of {\em cautious reasoning}: given a database $D$, a set of integrity constraints, and a query $q$, one wishes to find the set of answers to $q$  that lie in {\em every} possible repair of the database (a.k.a.\ {\em consistent answers}). This problem has been studied both in relational models and in graph databases, for usual repair notions such as set-based repairs, and by restricting the queries and classes of constraints to classes extensively studied in data exchange and data integration~\cite{barcelo2017data,Cate:2015}.  
We extend our study of repair 
prioritization to the computation of consistent answers under the same setting. 

This paper is organized as follows. In Section~\ref{sec:preliminaries} we introduce the necessary preliminaries and notation for the syntax and semantics for our data-graph model, as well as the definitions of consistency and prioritization. We define prioritized subset repairs for data-graphs, and present the different problems that can be studied in this setting.
In Section~\ref{sec:repairs} we study the computational complexity of determining the existence of a repair, repair checking and computing repairs for different preference criteria and fragments of \Gregxpath. Analogously, in Section~\ref{sec:cqa} we present a comprehensive study of the complexity of CQA in this setting, deferring
details on technical proofs to the Appendix provided in the supplementary material.
Conclusions and future work directions are discussed in Section~\ref{sec:conclusions}.

\section{Preliminaries}\label{sec:preliminaries}

Let $\Sigma_e$ be a fixed, non-empty set of edge labels, and $\Sigma_n$ a countable (either finite or infinite enumerable), non-empty set of data values (sometimes referred to as data labels), such that $\Sigma_e \cap \Sigma_n = \emptyset$. A \defstyle{data-graph}~$G$ is a tuple $(V,L,\dataFunction)$ where $V$ is a set of nodes, $L$ is a mapping from $V \times V$ to $\parts{\Sigma_e}$ defining the edges of the graph, and $\dataFunction$ is a function mapping the nodes from $V$ to data values from $\Sigma_n$. We denote by $E$ the set of edges of the data-graph, i.e. $E = \{(v,\aLabel,w) : v,w \in V, \aLabel \in L(v,w)\}$.
The \defstyle{cardinality} of a data-graph $\aGraph$, denoted by $|\aGraph|$, is given by $|V|+|E|$.

A data-graph $\aGraph = (V,L,D)$ is a \defstyle{subset} of a data-graph $\aGraph'=(V',L',D')$ (denoted by $\aGraph \subseteq \aGraph'$) if and only if $V \subseteq V'$, and for any $v,v' \in V$, we have $L(v,v') \subseteq L'(v,v')$ and $D(v) = D'(v)$. In this case, we also say that $\aGraph'$ is a \defstyle{superset} of $\aGraph$.
We denote the \textit{facts} related to a data-graph $G$ by $G^f = V \cup E$; and when it is clear from the context, we will omit the superscript $f$.
Sometimes we will refer to the nodes of a data-graph $G$ as $V_G$. If $V' \subseteq V$ we define the data-graph induced by $V'$ given $G$ as $G' = (V', L|_{V' \times V'}, D|_{V'})$. For simplicity, we might refer to a node using its data-value under the reasonable assumption of unique names.

%\nina{agregar def de inclusion de data graphs, intersecction de data graphs, y cardinality de un dg}

\defstyle{Path expressions} of \Gregxpath are given by:
$$ \aPath = \varepsilon \mid \labelComodin \mid \aLabel \mid \aLabel^- \mid \expNodoEnCamino{\aFormula} \mid \aPath \cdot\aPath \mid \aPath\cup\aPath \mid \aPath\cap\aPath \mid \aPath^* \mid \pathComplement{\aPath} \mid \aPath^{n,m} $$
% \begin{center}
%     $\aPath, \aPathb$ = $\varepsilon$ $|$ $\labelComodin$ $|$ $\aLabel$ $|$ $\aLabel^{-}$ $|$ $\expNodoEnCamino{\aFormula}$ $|$ $\aPath$ . $\aPathb$ $|$ $\aPath \pathUnion \aPathb$ $|$ $\aPath \pathIntersection \aPathb$ $|$ $\aPath^{*}$ $|$ $\pathComplement{\aPath}$ $|$
%     $\aPath^{n,m}$ 
% \end{center} 
where $\aLabel$ iterates over $\Sigma_e$, $n,m$ over $\mathbb{N}$ and $\aFormula$ is a \defstyle{node expression} defined by the following grammar:
$$ \aFormula = \neg \aFormula \mid \aFormula \wedge \aFormula \mid 
    \aFormula \vee \aFormula  \mid  \esDatoIgual{\aData} \mid \esDatoDistinto{\aData} \mid \comparacionCaminos{\aPath} \mid \comparacionCaminos{\aPath = \aPathb} \mid  \comparacionCaminos{\aPath \neq \aPathb} $$
% \begin{center}
%     $\aFormula, \aFormulab$ = $\neg \aFormula \mid \aFormula \wedge \aFormulab$ $|$ $\comparacionCaminos{\aPath}$ $|$ $\esDatoIgual{\aData}$ $|$ $\esDatoDistinto{\aData}$ $|$ $\comparacionCaminos{\aPath = \aPathb}$ $|$ $\comparacionCaminos{\aPath \neq \aPathb}$ $|$
%     $\aFormula \vee \aFormulab$
% \end{center}
%
where $\aPath$ and $\aPathb$ are path expressions (i.e. path and node expressions are defined by mutual recursion), and $c$ iterates over $\Sigma_n$. 
The size of path expressions $\aPath$ (resp.\ node expressions $\aNodeExpression$) is denoted by $|\aPath|$ (resp.\ $|\aNodeExpression|$), and is defined as the number of symbols in $\aPath$ (resp.\ $\aNodeExpression$), or equivalently as the size of its parsing tree. %  (resp.\ $\aNodeExpression$).
% (equivalently, we could define them as the sizes of parse trees of those expressions)
%\nina{how is the "size" of the expression R formally defined}
The semantics of these languages are defined  in \cite{libkin2016querying} similarly as the usual regular languages for navigating graphs \cite{barcelo2013querying}, while adding some extra capabilities such as the complement of a path expression $\pathComplement{\aPath}$ and data tests. The $\comparacionCaminos{\aPath}$ operator is the usual one for \textit{nested regular expressions} (NREs) used in \cite{barcelo2012relative}. 
Given a data-graph $\aGraph=(V,L,D)$, the semantics of \Gregxpath expressions are:
\footnotesize
\begin{align*}
    %Node expressions
    %Data values
    %Double condition
    \semantics{\esDatoIgual{c}}_G \eqdef{}& \set{u\in V \mid D(u)=\esDato{c}} 
    \\ %\qquad\qquad\qquad
    \semantics{\esDatoDistinto{c}}_G \eqdef{}&  V\setminus \semantics{\esDatoIgual{c}}_G\\
    %Logic operators
    \semantics{\varphi\land\psi}_G \eqdef{}&  \semantics{\varphi}_G\cap\semantics{\psi}_G 
    \\ %\!\qquad\qquad\qquad\qquad 
    \semantics{\varphi\lor\psi}_G \eqdef{}&  \semantics{\varphi}_G\cup\semantics{\psi}_G \\
    %Projection & Negation
    \semantics{\tup{\alpha}}_G \eqdef{}&  \set{u\in V\mid \exists v\in V.(u,v)\in\semantics{\alpha}_G} 
    \\ %\quad 
    \semantics{\lnot\varphi}_G \eqdef{}&  V\setminus \semantics{\varphi}_G \\
    %Path expressions
    %Epsilon & floor
    \semantics{\varepsilon}_G \eqdef{}&  \{(u,u) \mid u\in V\} 
    \\ %\qquad\qquad\qquad\quad\:\:\:\,
    \semantics{\labelComodin}_G \eqdef{}& \{(u,v) \mid L(u,v)\neq\emptyset\} \\
    %Atomic edge labels & inverses
    \semantics{\aLabel}_G \eqdef{}&  \set{(u,v)\in V^2\mid \aLabel \in L(u,v)} 
    \\ %\qquad\,\,\,\, 
    \semantics{\aLabel^-}_G \eqdef{}&  \set{(u,v)\in V^2\mid \aLabel \in L(v,u)} \\
    %Union, interception, concatenation & complement
    \semantics{\alpha\star\beta}_G \eqdef{}&  \semantics{\alpha}_G\star\semantics{\beta}_G \text{ for }\star \in \set{\circ,\cup,\cap} 
    \\ %\quad\:\:\:\:\:
    \semantics{\overline{\alpha}}_G \eqdef{}& V^2 \setminus \semantics{\alpha}_G \\
    %Test 
    \semantics{\expNodoEnCamino{\aFormula}}_G \eqdef{}&  \set{(u,u) \mid u\in V,u\in\semantics{\varphi}_G} 
    \\ %\qquad\, 
    \textstyle\semantics{\alpha^{n,m}}_G \eqdef{}& \bigcup_{k=n}^m (\semantics{\alpha}_G)^k \\
    %Kleene star
    \semantics{\alpha^*}_G \eqdef{}&  \text{the reflexive transitive closure of }\semantics{\alpha}_G.\\
    \semantics{\tup{\alpha = \beta}}_G \eqdef{}&  \set{u\in V\mid \exists u_1,u_2\in V, (u,u_1)\in\semantics{\alpha}_G, \\&(u,u_2)\in\semantics{\beta}_G, D(u_1)=D(u_2)}\\
    %Double condition negated
    \semantics{\tup{\alpha \neq \beta}}_G \eqdef{}&  \set{u\in V\mid \exists u_1,u_2\in V,
    (u,u_1)\in\semantics{\alpha}_G, \\& (u,u_2)\in\semantics{\beta}_G, D(u_1)\neq D(u_2)}
\end{align*}
\normalsize

For the sake of legibility of path expressions, for any $\aLabel \in \Sigma_e$ we will often use the notation $\down_\aLabel$ rather than simply $\aLabel$. %\sergio{Falta aclarar qué son $\entoncesNodo$ y $\entoncesCamino$, no? Después se usan sin explicar}
We also denote by $\aPath \entoncesCamino \aPathb$ the path expression $\aPathb \pathUnion \pathComplement{\aPath}$, and with $\aNodeExpression \entoncesNodo \aNodeExpressionb$ the node expression $\aNodeExpressionb \lor \neg \aNodeExpression$.
The positive fragment of $\Gregxpath$, denoted by $\Gposregxpath$, is obtained by removing the productions $\pathComplement{\aPath}$ and $\lnot \varphi$. Moreover, we denote by $\GNodeXPath$ the set of all node expressions from \Gregxpath{}, and \GNodePosXPath{} the intersection $\Gposregxpath \cap \GNodeXPath{}$ (i.e. the set of all positive node expressions).

\begin{definition}[Consistency]\label{def:consistency_global}
Let $\aGraph$ be a data-graph and $\aRestrictionSet = \aRestrictionSetPaths \union \aRestrictionSetNodes \subseteq \Gregxpath$ a finite set of restrictions, where $\aRestrictionSetPaths$ and $\aRestrictionSetNodes$ consist of path and node expressions, respectively. 
We say that $(\aGraph, \aRestrictionSet)$ is \defstyle{consistent} with respect to\ $\aRestrictionSet$, noted as $\aGraph \models \aRestrictionSet$, if the following conditions hold: 
\begin{itemize}
    \item $\forall \aNode \in V_\aGraph$ and $\aNodeExpression \in \aRestrictionSetNodes$, we have that $\aNode \in \semantics{\aNodeExpression}$ 
    \item $\forall \aNode,\aNodeb \in V_\aGraph$ and $\aPath \in \aRestrictionSetPaths$, we have that $(\aNode,\aNodeb) \in \semantics{\aPath}$
\end {itemize}
Otherwise we say that $\aGraph$ is inconsistent \wrt $\aRestrictionSet$. 
\end{definition}

This consistency notion can be evaluated in polynomial time: 

\begin{lemma}[\cite{libkin2016querying}, Theorem 4.3]\label{lemma:evaluate_consistency_polynomial}
    Given a data-graph $G$ and an expression $\eta \in \Gregxpath$, there is an algorithm that computes the set $\semantics{\eta}_G$ in polynomial time on the size of $G$ and $\eta$.
    As a consequence, we can decide whether $G \models R$ in polynomial time on the size of $G$~and~$R$.
\end{lemma}

\begin{example}
Consider the film database depicted in Figure \ref{figure:movieDB}, where the nodes represent people from the film industry (such as actors or directors) or movies.

\begin{figure}[ht]
    \centering
    \scalebox{0.85}{
    \begin{tikzpicture}[node distance={25mm}, thick, main/.style = {draw, rectangle}, scale=0.65, every node/.append style={transform shape}]
        \node[main] (Hoffman) {Hoffman};
        \node[main] (Actor) [left of=Hoffman] {Actor};
        \node[main] (Phoenix) [below left of=Actor] {Phoenix};
        \node[main] (Robbie) [below right of=Hoffman] {Robbie};
        \node[main] (Babylon) [above right of=Robbie] {Babylon};
        \node[main] (The Master) [above of=Hoffman] {The Master};
        \node[main] (film) [right of=Babylon] {Film};
        \node[main] (Anderson) [above of=film] {Anderson};
        \node[main] (Chazelle) [right of=Anderson] {Chazelle};
        
        \draw[->] (Hoffman) -- (Actor) node[midway, above=0.2pt, sloped]{type};
        
        \draw[->] (Phoenix) -- (Actor) node[midway, above=0.1pt, sloped] {type};
        
        \draw[->] (Robbie) -- (Actor) node[midway, above=0.2pt, sloped] {type};

        \draw[->, bend right=40] (Actor) edge  node[midway, above=0.2pt, sloped] {type} (Robbie);
        
        %\draw[->] (Actor) edge node[midway, above=1pt, sloped] {acts\_in} (The Master);
        
        \draw[->, bend left=50] (Phoenix) edge  node[midway, above=0.1pt, sloped] {acts\_in} (The Master) ;
        
        \draw[->] (Hoffman) -- (The Master) node[midway, sloped, above=0.1pt]{acts\_in};
        
        \draw[->] (The Master) -- (Anderson) node[midway, above=0.1pt] {directed\_by};
        
        \draw[->] (Babylon) -- (Chazelle) node[midway, above=0.1pt, sloped] {directed\_by};
        
        \draw[->] (Robbie) -- (Babylon) node[midway, above=0.1pt, sloped] {acts\_in};
        
        \draw[->] (Babylon) -- (film) node[midway, below=0.1pt] {type};
        
        \draw[->, bend right=10] (The Master) edge node[midway, above=0.1pt, sloped] {type} (film) ;
        
        \draw[->, bend left=20] (The Master) edge node[midway, above=0.1pt, sloped] {directed\_by} (Chazelle);
    \end{tikzpicture}
    }
    \caption{A film data-graph.}
    \label{figure:movieDB}
\end{figure}
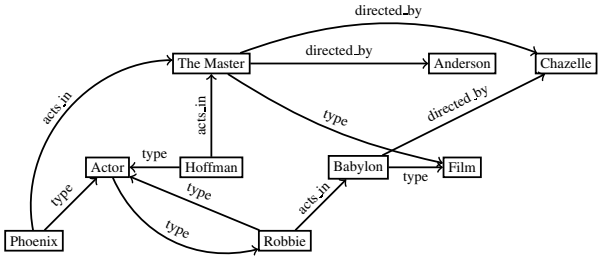
 
If we want to extract a subgraph that preserves only actors who have worked with (Philip Seymour) Hoffman in a film directed by (Paul Thomas) Anderson, then we want the following formula to be satisfied:
\footnotesize
\begin{align*}
 \aNodeExpression &= \comparacionCaminos{\down_\esLabel{type}[\esDato{actor}^=]}   \entoncesNodo \\ 
 &  \comparacionCaminos{\down_\esLabel{acts\_in}[\comparacionCaminos{\down_\esLabel{directed\_by} [\esDato{Anderson}^=]}] \down_\esLabel{acts\_in}^- [\esDato{Hoffman}^=]}
\end{align*}

\normalsize
Since Robbie did not work with Hoffman in a film directed by Anderson, $\aNodeExpression$ is not satisfied in the film data-graph, hence we do not have consistency \wrt $\{ \aNodeExpression \}$. Observe that the restriction also applies to Hoffman, thus it is required that he is featured in at least one film by Anderson to satisfy the constraint $\aNodeExpression$. 

Other reasonable restrictions are that nodes with ``data types'' such as $\esDato{Actor}$ should not have outgoing edges, or that nodes representing films should have an outgoing edge indicating who directed it. The first constraint can be expressed as
\(
\psi_1 = \esDatoIgual{\esDato{Actor}} \vee \esDatoIgual{\esDato{Film}} \entoncesNodo \lnot \comparacionCaminos{\down_\labelComodin},
\)
while the second one can be expressed by 
\(
\psi_2 = \comparacionCaminos{ \down_\esLabel{type}[\esDato{Film}^=] } \entoncesNodo \comparacionCaminos{\down_{\esLabel{directed\_by}}}.
\)
Observe that the first constraint is not satisfied by the data-graph as $\esDato{Actor} \notin \semantics{\psi_1}$, while $\psi_2$ is satisfied. Another possible restriction is that every film should have a unique director, which can be written as $\psi_3 = \comparacionCaminos{\down_\esLabel{type} [\esDato{Film}^=] } \entoncesNodo \lnot \comparacionCaminos{\down_\esLabel{directed\_by} \neq \down_{\esLabel{directed\_by}} }$, and is not satisfied by the node representing the film \esDato{The Master}. %\looseness=-1
\end{example}

\subsection{Types of preferences}\label{sec:types-preferences}
We focus on the problem of finding a consistent subset of a data-graph, introducing different minimality notions defined via \emph{prioritizations} given by preorders.
A \emph{preorder $\preccurlyeq$ over a set $S$} is a reflexive and transitive relation in $S\times S$.
Given a preorder $\preccurlyeq$ over data-graphs and two data-graphs $G_1,G_2$, we note $G_1 \prec G_2$ if $G_1 \preccurlyeq G_2$ and it is not the case that $G_2 \preccurlyeq G_1$.
The notion of prioritized subset repair in data-graphs can be defined as follows:
%\santi{Me parece que no deberiamos hablar de \textbf{subset} repairs, por lo menos a nivel de las definiciones tecnicas. Lo digo porque lo que fuerza al repair a ser un subconjunto es el criterio de preferencias en si (es mas, el simbolo $\subseteq$ no se esta usando en ningun lado de la definicion xD), asi que no ganamos nada agregando esa palabra. De paso simplificamos notacion. El trabajo de Molinaro lo hace de esta forma. Igual habria que aclarar que todos los criterios de preferencia que estudiamos fuerzan a los repairs a ser subconjuntos del grafo original.}
%\nina{Es cierto que decir 'subset repair' se refiere a una semantica de repacacion (buscar un subconjunto+maximal por set inclusion). Pero no es verdad en nuestro caso que estamos forzando la inclusion a traves de las preferencias.}
%\nina{adaptar esta def para que incluya superset, con resp. entre parentesis}
\begin{definition}[Prioritized subset repairs]
Let $\aRestrictionSet$ be a set of restrictions and $\aGraph$ a data-graph. We say that $\aGraph'$ is a \defstyle{prioritized subset repair} of $\aGraph$, denoted by \defstyle{\prefrepair{\preccurlyeq}} %or $\subseteq$-repair 
if: 
\begin{itemize}
    \item $(\aGraph',\aRestrictionSet)$ is \defstyle{consistent} (i.e. $\aGraph' \models \aRestrictionSet$).    
    \item $\aGraph' \subseteq \aGraph$.        
   \item There is no data-graph $\aGraph''\neq\aGraph'$ such that $(\aGraph'',\aRestrictionSet)$ is consistent and $ \aGraph' \prec \aGraph'' \preccurlyeq \aGraph$.
\end{itemize}
We denote the set of prioritized subset repairs of $\aGraph$ with respect to $\aRestrictionSet$ and preorder $\preccurlyeq$ as $\prefRepairsSet{\preccurlyeq}{G}{R}$.
\end{definition}

Given a data-graph $G=(V, L, D)$, a \defstyle{structural prioritization} (or simply a prioritization) $P = (P_1,...,P_\ell)$ of $G$ is a partition of $G$ into priority levels $P_i$, where $P_i \subseteq V \cup E$ for $1\leq i \leq n$. In this prioritization, $P_1$ contains the most reliable facts, whilst $P_\ell$ contains the least reliable facts of $G$.
%On the other hand, a \defstyle{data-prioritization} $P_d= (P_1,...,P_n)$ is a partition of $\Sigma_e \sqcup \Sigma_n$.

%Given a prioritization $P = (P_1,...,P_n)$ of the data-graph $G=(V, L, D)$, that is, $P$ is a partition of $G$ into the priority levels $P_i$, where $P_1$ contains the most reliable facts, and $P_n$ contains the least reliable facts of $G$.
%
%Given a function $\weight: \Sigma_e \sqcup \Sigma_n \to \N$ (where $\sqcup$ denotes disjoint union), we extend $w$ to any finite data-graph $G = (V,L,D)$ over $\Sigma_e$ and $\Sigma_n$ as
%\(
%\weight(G) = \sum_{x,y \in V} \left(\sum_{z \in L(x,y)} \weight(z) \right) + \sum_{x \in V} \weight(D(x)).
%\)
%
%
We study the following preorders, already considered in~\cite{DBLP:conf/kr/LukasiewiczMM23}:

\begin{enumerate}
    \item \emph{Inclusion} ($\subsetInclusionCriteria$): for every $G_1, G_2 \subseteq G$, $G_2$ is preferred to $G_1$ if $G_1 \subset G_2$. 
    \item \emph{Prioritized Set Inclusion} ($\subseteq_P$): Given a prioritization $P = (P_1,\ldots,P_\ell)$, for every $G_1, G_2 \subseteq G$, we write $G_1 \subseteq_P G_2$ iff for every $1\leq i\leq \ell$, $G_1\cap P_i = G_2\cap P_i$, or there is some $i$ such that $G_1^f\cap P_i \subsetneq G_2^f\cap P_i$ and for every $1\leq j<i$, $G_1^f\cap P_j=G_2^f\cap P_j$.
    \item \emph{Cardinality} ($\cardinalityCritera$): for every $G_1, G_2 \subseteq G$, $G_1 \cardinalityCritera G_2$ iff $|G_1| \leq |G_2|$.
    %\santi{Por que usamos $D_1,D_2$ y $D$ aca abajo en vez de $G_1,G_2$ y $G$?}
    \item \emph{Prioritized cardinality ($\leq_P$)}: For every $G_1, G_2 \subseteq G$, $G_1 \leq_P G_2$ iff for every $1 \leq i\leq \ell$, $|G_1^f \cap P_i| = |G_2^f \cap P_i|$, or there is some $i$ such that $|G_1^f \cap P_i|< |G_2^f \cap P_i|$ and for every $1\leq j<i$, $|G_1^f\cap P_j|=|G_2^f\cap P_j|$.
    \item \emph{Weights} ($\weightCriteria$):
    Given a weight function $\weight \colon \Sigma_e \sqcup \Sigma_n \to \N$ that assigns weights to all elements from $\Sigma_e \sqcup \Sigma_n$, we extend $w$ to any data-graph $G = (V,L,D)$ in the natural way as $w(G) = \sum_{x \in V} w(D(x)) + \sum_{x,y \in V} \left( \sum_{z \in L(x,y)} w(z)\right)$. Then, for every $G_1, G_2 \subseteq G$, $G_1 \leq_w G_2$ iff $\weight(G_1) \leq \weight(G_2)$. We assume that $w(x)$ can be computed in polynomial time on $|x|$, and moreover that $w(x) = O(2^{(|x|)})$. As a consequence, $\weight(G) = O(2^{(|G|)})$.
\end{enumerate}
We consider another preference criteria based on \textit{multisets}. 
\begin{definition}[Multisets]
    Given a set $A$, its set of \defstyle{finite multisets} is defined as $\multiset{A} =\{ M: A \to \N \mid M(x) \neq 0 \text{ only for a finite number of }x \}.$
    Given a partial order\footnote{That is, an antisymmetric, reflexive, and transitive relation} $(A, \le)$, the \defstyle{multiset ordering} $(\multiset{A}, \lessEqualMultiset)$ is defined as in {\em\cite{DershowitzManna,huet1980equations}}: $M_1 \lessEqualMultiset M_2$ iff either $M_1 = M_2$ or $M_1 \neq M_2$ and for all $x \in A$, if $M_1(x) > M_2(x)$, then there exists some $y \in A$ such that $x < y$ and  $M_1(y) < M_2(y)$.
\end{definition}
%
% \begin{example} Over $\multiset{\N}$ we have $\{0,0,0,1,1,1,2\} \lessMultiset \{2,2\}$. \end{example}
If $(A, \le)$ is a partial (resp.\ total) order, then $(\multiset{A}, \lessEqualMultiset)$ is a partial (resp.\ total) order. If $(A, \le)$ is a well-founded order\footnote{I.e. for all $S \subseteq A$, if $S \neq \emptyset$ then there exists $m \in S$ such that $s\not < m$ for every $s \in S$.}, then  we have that $(\multiset{A}, \lessEqualMultiset)$ is also a well-founded order~\cite{DershowitzManna}. We may note the elements of a multiset as pairs consisting of an element $a$ of $A$ and its multiplicity $k$, e.g.\ $(a, k)$.
\begin{definition}
    Given a finite data-graph $G$ over $\Sigma_e$ and $\Sigma_n$, we define its \defstyle{multiset of edges and data values} as the multiset \edgeDataMultiset{G} over $\Sigma_e$ and $\Sigma_n$ such that: 
    \[
    \edgeDataMultiset{G}(x) = \begin{cases}
    |\semantics{x}_\aGraph| & x \in \Sigma_e \\
    |\semantics{x^=}_\aGraph| & x \in \Sigma_n.  
    \end{cases}
    \]
\end{definition}
All these multisets of edges and data values belong to $\multiset{A}$, for $A = \Sigma_e \sqcup \Sigma_n$. The last preference criterion we consider for repairing is:
\begin{enumerate}
    \setcounter{enumi}{5}
    \item \textit{Multiset} ($\multisetCriteria$): Given a partial order $(\Sigma_e \cup \Sigma_n, \le)$, for every $G_1,G_2 \subseteq G$, $G_1 \multisetCriteria G_2$ iff $G_1^\mathcal{M} \leq_{mset} G_2^\mathcal{M}$. We assume that deciding whether $G_1 \preccurlyeq_{\mathcal{M}} G_2$ can be done in $O(poly(|G_1| + |G_2|))$.
\end{enumerate}

\begin{example}
In Example~\ref{figure:movieDB}, if $\varphi$ is the only integrity constraint imposed, by deleting the edge $(\esDato{Robbie}, \esLabel{type}, \esDato{Actor})$ we get a subset repair (that is, where the preorder is given by set inclusion) of the data-graph that is also a cardinality subset repair. If both $\varphi$ and $\psi_1$ are imposed then, to obtain a subset repair, it is necessary to delete the edges $(\esDato{Robbie}, \esLabel{type}, \esDato{Actor})$ and $(\esDato{Actor}, \esLabel{type}, \esDato{Robbie})$. Notice that deleting the node \esDato{Robbie} would also be enough to satisfy the constraints, however this does not yield a maximal consistent subset.

%\santi{Me parece que la ultima oracion es falsa, porque el grafo satisface $\psi_1$. Creo que esto quedo viejo de cuando teniamos las otras reglas (la regla era que todo nodo debe tener un tipo).}
%\nina{Si, en algun momento cambiamos el ejemplo de nuevo, ya que modificamos el grafo justamente para que no satisfaga $\psi_1$ cuando escribimos esa regla. Podemos sumar un eje (Actor,type, Robbie), en ese caso el subset repair es el mismo que con $\varphi$ + borrar ese nuevo eje. En cualquier caso, borrar el nodo Robbie no seria minimal.}

Suppose instead that $\aRestrictionSet = \{ \psi_2, \psi_3 \}$. In that case, the possible subset repairs are obtained by removing one of the outgoing \esLabel{directed\_by} edges from \esDato{The Master}. If we are given a prioritization $(P_1,P_2)$ such that $(\esDato{The Master}, \esLabel{directed\_by}, \esDato{Anderson}) \in P_1$ and $(\esDato{The Master}, \esLabel{directed\_by}, \esDato{Chazelle}) \in P_2$, then any ambiguity is removed and the edge that should be deleted is the one directed to \esDato{Chazelle}.
\end{example}

We now show an alternative running example that represents a physical network by means of a data-graph.

\begin{example}\label{example:WeightsAndRepairs}
Consider the following data-graph representing a physical network, where the edges are labeled according to two different quality levels of connection (e.g. varying robustness, resistance to physical attacks) which we call $\downarrow_\esLabel{low}$ and $\downarrow_\esLabel{high}$. 
Let $\aRestrictionSet = \{\aPath_{conn\_dir}, \aPath_{2l\rightarrow good}, \aPath_{hll\rightarrow rn}\}$ be a set of restrictions such that:
\begin{align*}
 \aPath_{conn\_dir} =& \labelComodin^*   \\
 \aPath_{2l\rightarrow good} =& \downarrow_\esLabel{low} \downarrow_\esLabel{low}
\entoncesCamino \downarrow_\esLabel{high} \downarrow_\esLabel{low} \pathUnion \downarrow_\esLabel{low} \downarrow_\esLabel{high} \\
& \pathUnion \downarrow_\esLabel{high} \downarrow_\esLabel{high} \pathUnion \downarrow_\esLabel{high} \pathUnion \downarrow_\esLabel{low} \\
 %&\aPath_{3l=bad} = ((((\downarrow_\esLabel{Low} \pathIntersection \pathComplement{\epsilon}) . (\downarrow_\esLabel{Low} \pathIntersection \pathComplement{\varepsilon})) \pathIntersection \pathComplement{\varepsilon}). (\downarrow_\esLabel{Low} \pathIntersection \pathComplement{\varepsilon})) \pathIntersection \pathComplement{\varepsilon})
%\pathIntersection (\downarrow_\esLabel{low}.(\downarrow_\esLabel{Low}.\downarrow_\esLabel{Low} \pathIntersection \pathComplement{\varepsilon}))
%\aPath_{hll} = & \downarrow_\esLabel{high} 
  \aPath_{hll\rightarrow rn} =& \downarrow_\esLabel{high} \downarrow_\esLabel{low} \downarrow_\esLabel{low} 
  \entoncesCamino (
  ((\downarrow_\esLabel{high} \pathIntersection \varepsilon) \downarrow_\esLabel{low} \downarrow_\esLabel{low}) \\ 
  \pathUnion& (\downarrow_\esLabel{high} (\downarrow_\esLabel{low} \pathIntersection \varepsilon) \downarrow_\esLabel{low}) 
  \pathUnion  \downarrow_\esLabel{high} \downarrow_\esLabel{low} (\downarrow_\esLabel{low} \pathIntersection \varepsilon))  
  \\
  %%%%
  \pathUnion & ((\downarrow_\esLabel{high} \downarrow_\esLabel{low} \pathIntersection \varepsilon) \downarrow_\esLabel{low})
  \pathUnion (\downarrow_\esLabel{high} (\downarrow_\esLabel{low} \downarrow_\esLabel{low} \pathIntersection \varepsilon)) 
  \\
  \pathUnion & (\downarrow_\esLabel{high} \downarrow_\esLabel{low} \downarrow_\esLabel{low} \pathIntersection \varepsilon)
  )
 %  \aPath_{hll=bad} =& ((((\downarrow_\esLabel{high} \pathIntersection \pathComplement{\varepsilon}) . (\downarrow_\esLabel{Low} \pathIntersection \pathComplement{\varepsilon})) \pathIntersection \pathComplement{\varepsilon}). (\downarrow_\esLabel{Low} \pathIntersection \pathComplement{\varepsilon})) \pathIntersection \pathComplement{\varepsilon}) \pathIntersection  \\
 %  &
 % \pathIntersection (\downarrow_\esLabel{high}.((\downarrow_\esLabel{Low}.\downarrow_\esLabel{Low}) \pathIntersection \pathComplement{\varepsilon})) \entoncesCamino  \varepsilon \pathIntersection \pathComplement{\varepsilon}
%  %\lnot \comparacionCaminos{\downarrow_\esLabel{low} \downarrow_\esLabel{low} \downarrow_\esLabel{low}}.
\end{align*}

$\aPath_{conn\_dir}$ expresses the notion of directed connectivity;  
$\aPath_{2l\rightarrow good}$  
establishes that, if a node can be reached by two low-quality edges, then it is also possible to reach it via a `good' path. That is, it can be reached either via only one edge (high or low), or in two steps but through at least one high-quality edge. 
Finally, $\aPath_{hll\rightarrow rn}$ enforces that, if we have a path $\downarrow_\esLabel{high} \downarrow_\esLabel{low} \downarrow_\esLabel{low}$ between two nodes, then there must exist a $\downarrow_\esLabel{high} \downarrow_\esLabel{low} \downarrow_\esLabel{low}$ path between the same nodes such that it revisits some node along its trajectory. 
%implies that we cannot have connections of length 3 that repeat no nodes and use only low-quality edges.

%Finally $\aPath_{hll=bad}$ means that any paths of the form $\downarrow_\esLabel{high} \downarrow_\esLabel{low} \downarrow_\esLabel{low}$ necessarily must repeat nodes\footnote{The antecedent of the formula defines a $\downarrow_\esLabel{high} \downarrow_\esLabel{low} \downarrow_\esLabel{low}$ path where no single step forms a loop, and also checks that the first node is not the third node, the first node is not the fourth node, and, lastly, that the second node is not the fourth node.}. \sergio{No, no hace esto. Pensar otra path expression y revisar los ejemplos}

We are instructed to make urgent changes to the network, and to do this we are given a weight function representing the cost of dismantling centrals (nodes) and connections in this network. This function assigns uniform weight to all data values $\weight(\aNode) = 20$, a cost for low-quality connections $\weight(\downarrow_\esLabel{low}) = 1$, and another for high-quality connections $\weight(\downarrow_\esLabel{high}) = 3$. %That is, building new high-quality connections is more expensive, but the maintenance is cheaper.
For a data-graph $\aGraph$ that does not satisfy the restrictions, a $\weight$-preferred subset repair can be interpreted as the most cost-effective way of obtaining a subset of the network that satisfies the restrictions while minimizing the dismantling costs given by $\weight$. For a full example, see Figure~\ref{figures:weightExamplePreferredRepair}.
%\vspace{-.25cm}
%\nina{hacemos un resize de las figuras o las reacomodamos como un minipage o algo de eso?}
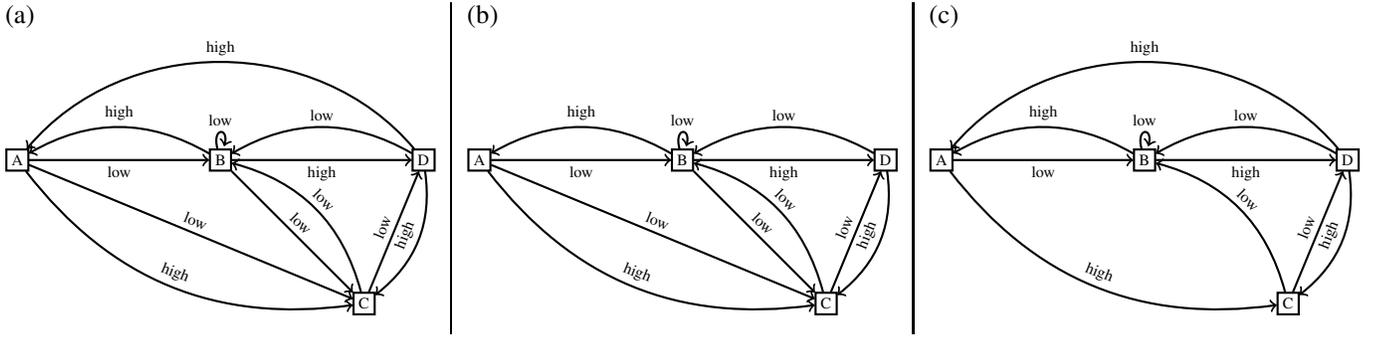
\begin{figure*}[ht]
	\centering
	\begin{tabular}{ l | l | l }
		(a) & (b) & (c) \\
 % \resizebox{7em}{7.5em}{%
        \begin{tikzpicture}%
			[rotate=0,node distance={45mm}, thick, main/.style = {draw, rectangle},  scale=0.6, every node/.append style={transform shape}] 
% 			\node[main] (1) {c}; 
% 			\node[main] (2) [below right of=1] {d}; 
% 			\node[main] (3) [left of=1] {b}; 
% 			\node[main] (4) [right of=1] {e};
% 			\draw[->] (3) -- (1) node[midway, below=0.5pt, sloped]{high};
% 			\draw[->, loop left] (3) edge node[midway, above=0.5pt, sloped]{low} (3); 
% 			\draw[->, bend right = 50] (4) edge node[midway, above=0.5pt, sloped]{low} (3); 
% 			\draw[->] (1) edge node[midway, below=0.5pt, sloped, pos=0.5]{low} (4);
% 			\draw[->, bend right = 30] (4) edge node[midway, above=0.5pt, sloped, pos=0.5]{high} (1); 
% 			\draw[->] (1) -- (2) node[midway, above=0.5pt, sloped, pos=0.5]{high};
% 			\draw[->, bend left = 30] (1) edge node[midway, above=0.5pt, sloped, pos=0.5]{low} (2); 
% 			\draw[->] (2) -- (3) node[midway, above=0.5pt, sloped, pos=0.5]{high};
% %   		\draw[->, bend right = 30] (3) edge node[midway, above=0.5pt, sloped, pos=0.5]{low} (2);
% 			\draw[->, bend right = 30] (1) edge node[midway, above=0.5pt, sloped, pos=0.5]{low} (3); 
% 			\draw[->] (2) -- (4) node[midway, above=0.5pt, sloped, pos=0.5]{low};
% 			\draw[->, bend right = 30] (1)  (3); 
   			\node[main] (1) {B}; 
			\node[main] (2) [below right of=1] {C}; 
			\node[main] (3) [left of=1] {A}; 
			\node[main] (4) [right of=1] {D};
			\draw[->] (3) -- (1) node[midway, below=0.5pt, sloped]{low};
			%\draw[->, loop left, looseness=30] (3) edge node[midway, above=0.5pt, sloped]{low} (3); %B
			\draw[->, bend right = 50] (4) edge node[midway, above=0.5pt, sloped]{high} (3); 
			\draw[->] (1) edge node[midway, below=0.5pt, sloped, pos=0.5]{high} (4);
			\draw[->, bend right = 30] (4) edge node[midway, above=0.5pt, sloped, pos=0.5]{low} (1); 
			\draw[->] (1) -- (2) node[midway, above=0.5pt, sloped, pos=0.5]{low};
			\draw[->, bend right = 30] (2) edge node[midway, above=0.5pt, sloped, pos=0.5]{low} (1); 
			\draw[->, bend right = 30] (3) edge node[midway, above=0.5pt, sloped, pos=0.5]{high} (2);
   			\draw[->] (3) -- (2) node[midway, above=0.5pt, sloped, pos=0.5]{low};
			\draw[->, bend right = 30] (1) edge node[midway, above=0.5pt, sloped, pos=0.5]{high} (3); 
			\draw[->] (2) -- (4) node[midway, above=0.5pt, sloped, pos=0.5]{low};
			\draw[->, bend right = 30] (1)  (3); 
      		\draw[->, bend left = 30] (4) edge node[midway, above=0.5pt, sloped, pos=0.5]{high} (2);
                \draw[->, loop above] (1) edge node[midway, above=0.5pt, sloped, pos=0.5]{low} (1); %new. Había un problema con la condición alpha2l-> good para la figura (b) sino, porque se podia ir a c y luego a d con low, sin mejor manera de hacer el camino

		\end{tikzpicture}
 %       }
        & %GRAFO (B)
		\begin{tikzpicture}%
			[rotate=0,node distance={45mm}, thick, main/.style = {draw, rectangle},  scale=0.6, every node/.append style={transform shape}] 
   			\node[main] (1) {B}; 
			\node[main] (2) [below right of=1] {C}; 
			\node[main] (3) [left of=1] {A}; 
			\node[main] (4) [right of=1] {D};
			\draw[->] (3) -- (1) node[midway, below=0.5pt, sloped]{low};
			%\draw[->, loop left, looseness=30] (3) edge node[midway, above=0.5pt, sloped]{low} (3);  %B
%			\draw[->, bend right = 50] (4) edge node[midway, above=0.5pt, sloped]{high} (3); 
			\draw[->] (1) edge node[midway, below=0.5pt, sloped, pos=0.5]{high} (4);
			\draw[->, bend right = 30] (4) edge node[midway, above=0.5pt, sloped, pos=0.5]{low} (1); 
			\draw[->] (1) -- (2) node[midway, above=0.5pt, sloped, pos=0.5]{low}; %Sacado porque se rompía 2L-> good
		\draw[->, bend right = 30] (2) edge node[midway, above=0.5pt, sloped, pos=0.5]{low} (1);%  Sacado porque se rompía 2L-> good :::: Pero perdemos conexion :(
			\draw[->, bend right = 30] (3) edge node[midway, above=0.5pt, sloped, pos=0.5]{high} (2);
   			\draw[->] (3) -- (2) node[midway, above=0.5pt, sloped, pos=0.5]{low};
			\draw[->, bend right = 30] (1) edge node[midway, above=0.5pt, sloped, pos=0.5]{high} (3); 
			\draw[->] (2) -- (4) node[midway, above=0.5pt, sloped, pos=0.5]{low};
			\draw[->, bend right = 30] (1)  (3); 
   			\draw[->, bend left = 30] (4) edge node[midway, above=0.5pt, sloped, pos=0.5]{high} (2);
                \draw[->, loop above] (1) edge node[midway, above=0.5pt, sloped, pos=0.5]{low} (1);%new
% 			\node[main] (1) {c}; 
% 			\node[main] (2) [below right of=1] {d}; 
% 			\node[main] (3) [left of=1] {b}; 
% 			\node[main] (4) [right of=1] {e};
% 			\draw[->] (3) -- (1) node[midway, below=0.5pt, sloped]{high};
% 			%\draw[->, loop left] (3) edge node[midway, above=0.5pt, sloped]{low} (3); 
% 			\draw[->, bend right = 50] (4) edge node[midway, above=0.5pt, sloped]{low} (3); 
% 			\draw[->] (1) edge node[midway, below=0.5pt, sloped, pos=0.5]{low} (4);
% 			\draw[->, bend right = 30] (4) edge node[midway, above=0.5pt, sloped, pos=0.5]{high} (1); 
% 			\draw[->] (1) -- (2) node[midway, above=0.5pt, sloped, pos=0.5]{high};
% %			\draw[->, bend left = 30] (1) edge node[midway, above=0.5pt, sloped, pos=0.5]{low} (2);
% 			\draw[->] (2) -- (3) node[midway, above=0.5pt, sloped, pos=0.5]{high};
% 			\draw[->] (2) -- (4) node[midway, above=0.5pt, sloped, pos=0.5]{low};
% 			\draw[->, bend right = 30] (1) edge node[midway, above=0.5pt, sloped, pos=0.5]{low} (3); 
		\end{tikzpicture}
        &  %GRAFO (C)
        \begin{tikzpicture}%
			[rotate=0,node distance={45mm}, thick, main/.style = {draw, rectangle},  scale=0.6, every node/.append style={transform shape}] 
      		\node[main] (1) {B}; 
			\node[main] (2) [below right of=1] {C}; 
			\node[main] (3) [left of=1] {A}; 
			\node[main] (4) [right of=1] {D};
			\draw[->] (3) -- (1) node[midway, below=0.5pt, sloped]{low};
			%\draw[->, loop left, looseness=30] (3) edge node[midway, above=0.5pt, sloped]{low} (3); % B
			\draw[->, bend right = 50] (4) edge node[midway, above=0.5pt, sloped]{high} (3); 
			\draw[->] (1) edge node[midway, below=0.5pt, sloped, pos=0.5]{high} (4);
			\draw[->, bend right = 30] (4) edge node[midway, above=0.5pt, sloped, pos=0.5]{low} (1); 
%			\draw[->] (1) -- (2) node[midway, above=0.5pt, sloped, pos=0.5]{low};
			\draw[->, bend right = 30] (2) edge node[midway, above=0.5pt, sloped, pos=0.5]{low} (1); 
			\draw[->, bend right = 30] (3) edge node[midway, above=0.5pt, sloped, pos=0.5]{high} (2);
%   		\draw[->] (3) -- (2) node[midway, above=0.5pt, sloped, pos=0.5]{low};
			\draw[->, bend right = 30] (1) edge node[midway, above=0.5pt, sloped, pos=0.5]{high} (3); 
			\draw[->] (2) -- (4) node[midway, above=0.5pt, sloped, pos=0.5]{low};
			\draw[->, bend right = 30] (1)  (3); 
   			\draw[->, bend left = 30] (4) edge node[midway, above=0.5pt, sloped, pos=0.5]{high} (2);
                \draw[->, loop above] (1) edge node[midway, above=0.5pt, sloped, pos=0.5]{low} (1);%new
% 			\node[main] (1) {c}; 
% 			\node[main] (2) [below right of=1] {d}; 
% 			\node[main] (3) [left of=1] {b}; 
% 			\node[main] (4) [right of=1] {e};
% 			\draw[->] (3) -- (1) node[midway, below=0.5pt, sloped]{high};
% 			\draw[->, loop left] (3) edge node[midway, above=0.5pt, sloped]{low} (3); 
% %			\draw[->, bend right = 50] (4) edge node[midway, above=0.5pt, sloped]{low} (3); 
% 			\draw[->] (1) edge node[midway, below=0.5pt, sloped, pos=0.5]{low} (4);
% 			\draw[->, bend right = 30] (4) edge node[midway, above=0.5pt, sloped, pos=0.5]{high} (1); 
% 			\draw[->] (1) -- (2) node[midway, above=0.5pt, sloped, pos=0.5]{high};
% 			\draw[->, bend left = 30] (1) edge node[midway, above=0.5pt, sloped, pos=0.5]{low} (2);
% 			\draw[->] (2) -- (3) node[midway, above=0.5pt, sloped, pos=0.5]{high};
% 			\draw[->, bend right = 30] (1) edge node[midway, above=0.5pt, sloped, pos=0.5]{low} (3); 
% 			\draw[->] (2) -- (4) node[midway, above=0.5pt, sloped, pos=0.5]{low};
% 			\draw[->, bend right = 30] (1)  (3); 
		\end{tikzpicture}
	\end{tabular}
	\caption{(a) A data-graph that does not satisfy the constraints from Example~\ref{example:WeightsAndRepairs} (for example, $(\esDato{B}, \esDato{D}) \notin \semantics{
 %\aPath_{hll=bad}
 \aPath_{hll\rightarrow rn}
 }$).
    (b) A subset data-graph (which is a subset repair) of the example of figure (a) that satisfies the constraints but is not a $\weight$-preferred subset repair; the associated weight of this repair is $w(G) - 3$ (from one \esLabel{high} edge).
    (c) A $\weight$-preferred subset repair; the associated weight of this repair is $w(G) - 2$ (from two deleted \esLabel{low} edges).
    %subset of the data-graph of figure (a) that satisfies $\aPath_{connected\_dir}$ and $\aPath_{2l\rightarrow good}$ but does not satisfy $\aPath_{3l=bad}$: there is a path of length 3 using only $\esLabel{low}$ edges. 
    }
	\label{figures:weightExamplePreferredRepair}
\end{figure*}
\end{example}
%    \santi{Creo que deberiamos agregar una explicacion de este ejemplo. Entiendo que uno de los pares que no cumplen de a) es $(\esDato{c}, \esDato{b})$. Despues, no entiendo por que hay que borrar el loop en \esDato{b}}
%    \sidenina{la regla que queria poner originalmente era 'no hay caminos L-H-L que repitan nodos' porque eso te da repairs que borran ejes H y otros que no y tenia mas sentido; pero lo cambie 8 veces y quedo mal... Luego lo charlamos}
%
\begin{example}[Cont. Example~\ref{example:WeightsAndRepairs}]
    %Consider the data-graphs depicted in Figure~\ref{figures:weightExamplePreferredRepair}.
    Ignoring all possible data values, the multisets corresponding to graphs (b) and (c) of Figure~\ref{figures:weightExamplePreferredRepair} are (with the informal multiset notation): $\{ (\esLabel{low}, 7), (\esLabel{high},4) \}$ and $\{ (\esLabel{low}, 5), (\esLabel{high},5) \}$. 
    %$\{\esLabel{low}, \esLabel{low}, \esLabel{low}, \esLabel{high},\esLabel{high}, \esLabel{high}, \esLabel{high}\}$ and $\{\esLabel{low}, \esLabel{low}, \esLabel{low}, \esLabel{low}, \esLabel{low}, \esLabel{high}, \esLabel{high},\esLabel{high}, \esLabel{high}\}$, respectively. 
    If we assume $\esLabel{low} < \esLabel{high}$, the data-graph (c) is \multisetCriteria-preferred to (b).
\end{example}

\subsection{Computational problems related to preferred repairs}

When it comes to repairing a database, there are several problems that we can define to obtain a more fine-grained and broad picture of the complexity of finding a consistent instance of the original database. Naturally, these problems can be translated to the data-graph repairing setting. In the sequel we define the repair problems that we will study in this work.
The intricacy of these problems usually lies in the set of allowed expressions $\mathcal{L}$ and $\queryLanguage$ for the constraints and the queries, respectively. Throughout this work, we will assume that these sets are always expressions in \Gregxpath.

Let us consider the \emph{empty graph} defined as $\aGraph = (\emptyset, L, D)$, from now on denoted by $\emptyset$. Since the empty graph satisfies every set of restrictions, every graph $\aGraph$ has a subset repair for any set $\aRestrictionSet$ of restrictions. %although this might not be the case when another (ad hoc) preference criteria is introduced. 
To better understand the complexity of the repairing task, we define the following decision problems:%\sideedwin{no están incompletos estos problemas? Porque si el criterio es Prioritized-set el grafo $G$ debería estar acompañado de un $P$ o si el criterio es weight debería haber también una función $w$. De dónde saldrían estas cosas si no vienen en el input?} \nina{si, lo pense tambien al escribirlo, pero en el paper de molinaro los definen asi: el problema es existencia de repair dado un preorden que viene definido sobre un conjunto S de dbs (que asumo es el espacio de repairs), asi que lo deje igual.}
%\edwin{Ahí estoy viendo, es un poco raro eso que hacen en el artículo porque en las demostraciones cuando dan una instancia la dan junto al $w$ o al $P$, o sea que estos se entienden como parte del input. Incluso en las demostraciones de Hardness construyen un prioritization que depende de la 3cnf. Me parece que tenemos que comentar algo al respecto.}
%
%
%
\begin{center}
\fbox{\begin{minipage}{22em}
  \textsc{Problem}: \problemRepairExistence{\preccurlyeq}

\textsc{Input}: A data-graph $\aGraph$ and a set $\aRestrictionSet \subseteq \mathcal{L}$.

\textsc{Output}: Decide if there exists $H \in \prefRepairsSet{\preccurlyeq}{G}{R}$ such that $H \neq \emptyset$.
\end{minipage}}
\end{center}

\begin{center}
\fbox{\begin{minipage}{22em}
  \textsc{Problem}: \problemRepairChecking{\preccurlyeq}

\textsc{Input}: Data-graphs $\aGraph$ and $\aGraph'$, and a set $\aRestrictionSet \subseteq \mathcal{L}$.

\textsc{Output}: Decide if $\aGraph' \in \prefRepairsSet{\preccurlyeq}{G}{R}$.
\end{minipage}}
\end{center}
We consider the functional version of the repair problem:
\begin{center}
\fbox{\begin{minipage}{22em}
  \textsc{Problem}: \problemRepairComputing{\preccurlyeq}

\textsc{Input}: A data-graph $\aGraph$ and a set $\aRestrictionSet \subseteq \mathcal{L}$.

\textsc{Output}: A data-graph $H \in \prefRepairsSet{\preccurlyeq}{G}{R}$.
\end{minipage}
}
\end{center}

Whenever $\preccurlyeq \in \{\priorityzedInclusionCriteria, \priorityzedCardinalityCriteria\}$ the input also contains the prioritization $(P_1,\ldots,P_\ell)$, which has size $O(|G|)$. For the cases $\preccurlyeq \in \{\weightCriteria, \multisetCriteria\}$ we will assume the function $w$ and the order $<$ to be fixed and polynomial-time computable, as specified in the definitions of the preorders. Considering both given as part of the input (in some reasonable codification) does not change the obtained results, since all upper bounds hold when these elements are fixed, while lower bounds hold under the tractability assumption.

Another well-known problem is CQA, where we want to find the \emph{consistent answers} of $q$ \wrt the database and a set of integrity constraints, that is, the set of all the answers to $q$ that lie in every possible repair of the original instance.
\begin{center}
\fbox{\begin{minipage}{22em}
  \textsc{Problem}: \problemCQAsubset{\preccurlyeq}

\textsc{Input}: A data-graph $\aGraph = (V,L,D)$, a set $\aRestrictionSet \subseteq \mathcal{L}$, a query $q \in \queryLanguage$ and two nodes $v,w \in V$.

\textsc{Output}: Decide whether $(v,w) \in \semantics{q}_H$ for all $H \in \prefRepairsSet{\preccurlyeq}{G} {R}$.
\end{minipage}}
\end{center}
In this paper, we perform a systematic study of the complexity of all the problems related to repairing a data-graph, for the notions of Inclusion, Cardinality, Prioritized set inclusion and Cardinality, Weights, and Multisets.

\paragraph{Complexity classes} Aside from the usual complexity classes \PTIME{}, \NP{} and \coNP{}, in this work we will use the following classes: \deltaptwo{} (i.e.,\ problems solvable in polynomial time using an \NP{} oracle), \deltaptwo{}[$\log n$] (same as before but using $O(\log |x|)$ calls to an \NP{} oracle on input $x$) and \piptwo{} (i.e.,\ the complement of the class of problems solvable in non-deterministic polynomial time with an \NP{} oracle, $\Sigma_p^2$). 

\section{Complexity of preferred repairs}\label{sec:repairs}

We start with the following lemma, which relates the subset inclusion criterion to all other preference criteria. Throughout this work, when we say $\preccurlyeq$ is {\em any preference criteria}, we mean any preference criteria from \S\ref{sec:types-preferences}

\begin{lemma}\label{obs:all_pref_rep_are_subset}
    Let $G$ be a data-graph, $R \subseteq \Gregxpath$, and $\preccurlyeq$ any preference criteria. Then, $\prefRepairsSet{\preccurlyeq}{G}{R} \subseteq \prefRepairsSet{\subsetInclusionCriteria}{G}{R}$.
\end{lemma}

\begin{proof}
    All preference criteria require maximality with respect to set inclusion:
    \begin{itemize}
        \item $\preccurlyeq\; =\; \cardinalityCritera$: if $H \subset H'$ then $|H| < |H'|$, and thus $H < H'$.
        \item $\preccurlyeq \; = \; \priorityzedInclusionCriteria$: if $H \subset H'$ then $H \cap P_i \subseteq H' \cap P_i$ for all sets $P_i$ and $H \cap P_j \subset H' \cap P_j$ for some $j$. Consequently $H \subset_P H'$.
        \item $\preccurlyeq = \weightCriteria$: this follows analogously as for $\cardinalityCritera$, using that $\weight$ assigns positive weights.
        \item $\preccurlyeq = \priorityzedCardinalityCriteria$: follows analogously as for $\priorityzedInclusionCriteria$.%
        \item $\preccurlyeq = \multisetCriteria$: if $H \subset H'$, then $H^{\mathcal{M}}(x) \leq H^{\mathcal{M}}(x)$ for all $x \in \Sigma_e \cup \Sigma_n$ and $H^{\mathcal{M}}(y) < H^{\mathcal{M}}(y)$ for some $y$. Thus, $H <_{\mathcal{M}} H'$. 
    \end{itemize}%
    \vspace{-0.4cm}
\end{proof}
We recall the following known result involving subset inclusion repairs for a set of restrictions that contains only node expressions from the positive fragment of $\Gregxpath$:

\begin{lemma}[\cite{abriola2023complexity}, Theorem 15]\label{lemma:unique_subset_repair_node_pos}
    Let $G$ be a data-graph and $R \subseteq \GNodePosXPath$. Then, $G$ has a unique $\prefrepair{\subsetInclusionCriteria}$ with respect to $R$, and it can be found in polynomial time.
\end{lemma}
Merging these two facts, we obtain the following corollary:
\begin{corollary}\label{coro:unique_repair_node_pos}
    Let $G$ be a data-graph, $R \subseteq \GNodePosXPath$ and $\preccurlyeq$ any preference criteria. Then, $G$ has a unique $\prefrepair{\preccurlyeq}$ with respect to $R$, which can be found in polynomial time.\looseness=-1
\end{corollary}
\begin{proof}
    By Lemma~\ref{lemma:unique_subset_repair_node_pos}, $G$ has a unique $\prefrepair{\subsetInclusionCriteria}$, which can be found in polynomial time. By Lemma~\ref{obs:all_pref_rep_are_subset} it follows that this unique $\prefrepair{\subsetInclusionCriteria}$ must also be $\prefrepair{\preccurlyeq}$: otherwise $\prefRepairsSet{\preccurlyeq}{G}{R}$ would be empty, which is absurd.
\end{proof}
We also establish expressivity relations between some preference criteria.
\begin{observation}\label{obs:trivial_relations_of_preferences}
    Let $G$ be a data-graph and $R \subseteq \Gregxpath$. Then, it holds that:
    \begin{itemize}
        \item $\prefRepairsSet{\subsetInclusionCriteria}{G}{R} = \prefRepairsSet{\priorityzedInclusionCriteria}{G}{R}$ when considering the trivial prioritization $P = \{G\}$.
        \item $\prefRepairsSet{\cardinalityCritera}{G}{R} = \prefRepairsSet{\weightCriteria}{G}{R} = \prefRepairsSet{\priorityzedCardinalityCriteria}{G}{R}$ when considering the trivial weight function $w(x) = 1$ and prioritization $P = \{G\}$.
    \end{itemize}
    As a consequence, any lower bound for the $\subsetInclusionCriteria$ criteria extends to the $\priorityzedInclusionCriteria$ criteria, while upper bounds for the latter one extend to the former. The same holds for $\cardinalityCritera$ with respect to both $\weightCriteria$ and $\priorityzedCardinalityCriteria$.\looseness=-1
\end{observation}
Finally, the following lemma will be useful to extend lower and upper bounds:
\begin{lemma}\label{lemma:existence_reduces_to_repair_checking}
    The problem \problemRepairExistence{\preccurlyeq} reduces to the complement of \problemRepairChecking{\preccurlyeq}. Moreover, the problem \problemRepairExistence{\subsetInclusionCriteria} reduces to \problemRepairExistence{\preccurlyeq} for all preference criteria $\preccurlyeq$.
    
    Therefore, a lower bound for \problemRepairExistence{\subsetInclusionCriteria} implies a lower bound for all \problemRepairExistence{\preccurlyeq} and \problemRepairChecking{\preccurlyeq} problems considering any $\preccurlyeq$.
\end{lemma}

\begin{proof}
    For the first part, note that $(G,R)$ is a positive instance of \problemRepairChecking{\preccurlyeq} if and only if $(G, \emptyset, R)$ is a negative instance of \problemRepairChecking{\preccurlyeq}.

    For the second statement, if $G$ has a non-empty \prefrepair{\subsetInclusionCriteria} $H$ with respect to $R$, then it must have some \prefrepair{\preccurlyeq} $H'$ such that $H \preccurlyeq H' \preccurlyeq G$. Finally, observe that for every proposed criteria it holds that, if $H \preccurlyeq H'$ and $H \neq \emptyset$, then $H' \neq \emptyset$.
    %\sideedwin{Esto se puede demostrar pero creo que no es cierto que $H\subseteq H'$ o $|H| \leq |H'|$.}
    %\santi{Tenes razon, hay que usar un argumento distinto para cada uno. Lo que si es cierto es que si $\emptyset$ no es un \prefrepair{\subsetInclusionCriteria}. Ahi lo arregle entonces tampoco es un \prefrepair{\preccurlyeq}, y con eso alcanza. Si te parece bien borra los comments.}
\end{proof}

\subsection{Upper bounds / Membership}

Let $\mathcal{L}$ be the language for the set of restrictions $R$. If $\mathcal{L} = \GNodePosXPath$, it follows from Corollary~\ref{coro:unique_repair_node_pos} that for each pair $(G,R)$ its unique \prefrepair{\preccurlyeq} can be computed in polynomial time, and consequently we can solve all the repair related problems.

\begin{proposition}\label{prop:all_repair_problems_ptime_node_pos}
    The problems \problemRepairExistence{\preccurlyeq}, \problemRepairChecking{\preccurlyeq} and \problemRepairComputing{\preccurlyeq} are in \PTIME{} for any preference criteria $\preccurlyeq$ and $\mathcal{L} \subseteq \GNodePosXPath$.
\end{proposition}

We obtain \coNP{} as an upper bound for \problemRepairChecking{\preccurlyeq} under arbitrary \Gregxpath constraints.

\begin{proposition}
    The problem \problemRepairChecking{\preccurlyeq} is in \coNP{}, for any preference criteria $\preccurlyeq$ and $\mathcal{L} \subseteq \Gregxpath$.
\end{proposition}

\begin{proof}
    $(G, G', R)$ is a negative instance of \problemRepairChecking{\preccurlyeq} if there exists a data-graph $G''$ such that $G' \prec G'' \preccurlyeq G$ and $G'' \models R$. Comparing two data-graphs according to $\preccurlyeq$ can be done in polynomial time by our assumptions over the preference criteria, and testing if $G''$ is consistent with respect to $R$ is in \PTIME due to Lemma~\ref{lemma:evaluate_consistency_polynomial}.
\end{proof}

By Lemma~\ref{lemma:existence_reduces_to_repair_checking} and this last result, we have an \NP{} upper bound for \problemRepairExistence{\preccurlyeq} for any subset of $\Gregxpath$.

\subsection{Lower bounds}

We start this section considering two known lower bound complexity results for \problemRepairExistence{\subseteq}:

\begin{proposition}[\cite{abriola2023complexity}, Theorem 11]\label{prop:repair_existence_hard_subset_path_pos}
    There exists a set of constraints $R \subseteq \Gposregxpath$ such that \problemRepairExistence{\subseteq} is \NPhard{}.
\end{proposition}

\begin{proposition}[\cite{abriola2023complexity}, Theorem 12]\label{prop:repair_existence_hard_subset_node}
    There exists a set of constraints $R \subseteq \GNodeXPath$ such that \problemRepairExistence{\subseteq} is \NPhard{}.
\end{proposition}

By Lemma~\ref{lemma:existence_reduces_to_repair_checking}, all repairing decision problems are \NPhard{} (except \problemRepairChecking{\preccurlyeq} which is \coNPhard{}) even when considering data complexity.

%\nina{cambié los nombres de los criterios por los simbolos y achique de 0.8 a 0.7, digan si les parece bien y cambiamos lo otro. Tambien quiza sacaria "Reg" de los nombres de los fragmentos y dejaria GXPath coso y ya. Quiza incluso mencionaria que podemos decirle de una manera o de otra en la intro.}
%\santi{Me parece bien! En el original de Libkin usaban \textit{core} para referirse al fragmento donde la estrella de Kleene solo puede aplicarse a atomos de la forma $\down_{\aLabel}$, mientras que el fragmento regular hacia uso arbitrario de Kleene. Estaria bueno aclararlo.}
\begin{table}
    \centering
    \renewcommand{\arraystretch}{1.1}
    \scalebox{0.8}{
    \begin{tabular}{|c|c|c|c|c|}
    \hline
        & \textbf{\GNodePosXPath} &  \textbf{\Gposregxpath}  & \textbf{\GNodeXPath} &  \textbf{\Gregxpath}  \\
        \hline
        \textbf{$\subsetInclusionCriteria$} & \PTIME{} & \NPcomplete{} &  \NPcomplete{} & \NPcomplete{} \\
        \hline        
        \textbf{$\cardinalityCritera$} &  \PTIME{} & \NPcomplete{} & \NPcomplete{} &  \NPcomplete{}\\
        \hline        
        \textbf{$\weightCriteria$} & \PTIME{} & \NPcomplete{} & \NPcomplete{} & \NPcomplete{} \\
        \hline        
        \textbf{$\priorityzedCardinalityCriteria$} & \PTIME{} & \NPcomplete{} & \NPcomplete{} & \NPcomplete{} \\
        \hline   
        \textbf{$\priorityzedInclusionCriteria$} & \PTIME{} & \NPcomplete{} & \NPcomplete{} & \NPcomplete{} \\   
        \hline
        \textbf{$\multisetCriteria$} & \PTIME{} & \NPcomplete{} & \NPcomplete{} & \NPcomplete{} \\
        \hline
    \end{tabular}
    }
    \caption{Complexity of the \problemRepairExistence{\preccurlyeq} problem for the different preference criteria and subsets of \Gregxpath. For \problemRepairChecking{\preccurlyeq} we obtain the same results but swapping all \NPcomplete{} entries for \coNPcomplete{}. All hardness results are obtained in data complexity, while the \PTIME{} results take into account the combined complexity.}
    \label{tab:repair_results}
\end{table}

In general, a problem in \NP is called \emph{self-reducible} if its function variant can be solved in polynomial time using an oracle to decide the original problem. Every \NP-complete problem is self-reducible. This gives the following straightforward result for the functional version for which we proved in the latter that the decision version is \NP-complete for every preference criteria: 
%It follows that the functional version is FNP-complete since we proved above that the decision version is NP-complete.

%\nina{ver url comentada}
%\url{https://en.wikipedia.org/wiki/Function_problem#:~:text=Self%2Dreducibility,-Observe%20that%20the&text=In%20general%2C%20a%20problem%20in,complete%20problem%20is%20self%2Dreducible.}

\begin{remark}\label{theo:setinc-functional}
The functional problem for every preference
relation is FNP-complete, when considering constraints in $\Gregxpath$, $\Gposregxpath$, or $\GNodeXPath$. % with the exception of $\GNodePosXPath$.
\end{remark}

\section{Complexity of CQA}\label{sec:cqa}

\begin{table}
    \centering
    \renewcommand{\arraystretch}{1.2}
    \scalebox{0.8}{
    \begin{tabular}{|c|c|c|c|c|}
    \hline
        & \textbf{\GNodePosXPath} &  \textbf{\Gposregxpath}  & \textbf{\GNodeXPath} &  \textbf{\Gregxpath}  \\
         \hline
        \textbf{$\subsetInclusionCriteria$}  & \PTIME{} & \piptwoComplete{} & \piptwoComplete{} & \piptwoComplete{}\\
         \hline   
         \textbf{$\cardinalityCritera$} & \PTIME{} & \deltaptwologComplete{} & \deltaptwologComplete{} & \deltaptwologComplete{}  \\
         \hline            
        \textbf{$\weightCriteria$}  & \PTIME{} & \deltaptwoComplete{} & \deltaptwoComplete{} & \deltaptwoComplete{}\\
         \hline        
       \textbf{$\priorityzedCardinalityCriteria$} & \PTIME{} & \deltaptwoComplete{} & \deltaptwoComplete{} & \deltaptwoComplete{}\\
         \hline        
        \textbf{$\priorityzedInclusionCriteria$} & \PTIME{} & \piptwoComplete{} & \piptwoComplete{} & \piptwoComplete{}\\
         \hline
         \textbf{$\multisetCriteria$} & \PTIME{} & \piptwoComplete{} & \piptwoComplete{} & \piptwoComplete{}
         \\
         \hline 
    \end{tabular}
    }
    \caption{Complexity of the \problemCQAsubset{} problem for the different preference criteria and subsets of \Gregxpath. The hardness results hold for the data complexity of the problem, while the tractable cases take into account the combined complexity.}
    \label{tab:cqa_results}
\end{table}

\subsection{Upper bounds / Membership}

Following our previous observations, we can conclude that:
\begin{proposition}\label{teo:cqa_ptime_node_pos}
    The problem \problemCQAsubset{\preccurlyeq} is in \PTIME{} for any preference criteria $\preccurlyeq$ if $\mathcal{L} \subseteq \GNodePosXPath$.
\end{proposition}
\begin{proof}
    By Corollary~\ref{coro:unique_repair_node_pos}, if $R \subseteq \GNodePosXPath$ then there is a unique \prefrepair{\preccurlyeq}, and it is possible to find it in \PTIME{}. Therefore, we can solve \problemCQAsubset{\preccurlyeq} by computing this unique repair and asking the query on it.\looseness=-1
\end{proof}
We now show upper bounds for the other cases. To do this, we obtain two different results: one for the \textit{inclusion-based} 
and \textit{multiset} criteria, and another for the \textit{cardinality-based} ones, using similar arguments to those in~\cite{Bienvenu-tech-report}.
    
\begin{theorem}\label{teo:upper_bounds_cqa}
    Let $\mathcal{L} \subseteq \Gregxpath$. Then, the problems \problemCQAsubset{\priorityzedInclusionCriteria} and \problemCQAsubset{\multisetCriteria} are in \piptwo, \problemCQAsubset{\cardinalityCritera} is in \deltaptwo$[\log n]$, and \problemCQAsubset{\weightCriteria} and \problemCQAsubset{\priorityzedCardinalityCriteria} are in \deltaptwo.
\end{theorem}
\begin{proof}
    Observe that $(G,R,q,v,w)$ is a negative instance of \problemCQAsubset{\priorityzedInclusionCriteria} if and only if:
    \begin{align*}
    &\exists \, G'\subseteq G. \; \forall \, G''\subseteq G. \\ 
    &[ G' \models R \wedge (v,w) \notin \semantics{q}_{G'} \wedge \left( G' \subsetneq_P G'' \implies G'' \nvDash R \right) ]
    \end{align*}
    This is a $\exists \forall$ formula, whose predicate can be evaluated in polynomial time as long as it is in \PTIME{} (1) to check consistency with respect to $R$, (2) to answer the query $q$, and (3) to check whether $G' \priorityzedInclusionCriteria G''$. We can test (1) and (2) in \PTIME{} due to Lemma~\ref{lemma:evaluate_consistency_polynomial}, while (3) follows from our hypothesis over the preference criteria. The same holds for $\multisetCriteria$.
    
    Now we show the \deltaptwo$[\log n]$ algorithm for \problemCQAsubset{\cardinalityCritera}. Given $G$ and $R$, we start by doing binary search in the range $[0,|G|]$ to obtain the amount of elements $z$ of the $\cardinalityCritera$-repairs. This is achieved by querying the oracle with the problem ``\textit{Is there a data-graph $H \subseteq G$ such that $H \models R$ and $|H| \geq z$?}'', which is in \NP{}. Once we find the biggest $z$, we perform a final \NP{} query of the form ``\textit{Is there a data-graph $H \subseteq G$ such that $H \models R$, $|H| = z$ and $(v,w) \notin \semantics{q}_H$?}''. The output of the algorithm is the negation of the answer from this last query. The same idea can be used to obtain a \deltaptwo$[n]$ algorithm for the case of $\weightCriteria$ by doing binary search over the range $[0, w(G)]$. The number of queries required follows from the fact that $w(G) = O(2^{|G|})$ by our hypothesis over the weight function.\looseness=-1

    Finally, the \deltaptwo$[n]$ algorithm for \problemCQAsubset{\priorityzedCardinalityCriteria} consists in applying the previous algorithm for all the different priority levels $P_1,\ldots, P_\ell$: we iteratively compute the optimal intersection $k_i$ between a repair and the set $P_i$. For each $1 \leq i \leq \ell$, we perform $|P_i|$ queries of the form ``\textit{Is there a data-graph $H \subseteq G$ such that $H \models R$, $|H \cap P_j| = k_j$ for $1 \leq j \leq i-1$, and $|H \cap P_i| = z$?}'' testing all possible values for $z$. Once we obtain all the values $k_1,\ldots, k_\ell$ we ask ``\textit{Is there a data-graph $H \subseteq G$ such that $H \models R$, $|H \cap P_j| = k_j$ for $1 \leq j \leq \ell$, and $(v,w) \notin \semantics{q}_H$?}'' and return the opposite of the answer to this last query.
\end{proof}
\subsection{Lower bounds}

We prove the \piptwo-\textit{hardness} of \problemCQAsubset{\subsetInclusionCriteria} for a fixed query $q$ and a fixed set of restrictions $R \subseteq \Gposregxpath$. By adapting the proof we also show hardness for a fixed set $R' \subseteq \GNodeXPath$.\looseness=-1 %The construction is inspired by the one in~\cite{barcelo2017data}[Theorem 1].
\begin{theorem}\label{teo:pi_p_2_hardness}
    There exists a set of constraints $R\subseteq \Gposregxpath$ and a query $q$ such that the problem \problemCQAsubset{\subsetInclusionCriteria} is \piptwo-\textit{hard}. Furthermore, there is a set of constraints $R' \subseteq \GNodeXPath$ that yields the same hardness results considering the query $q$.
\end{theorem}
\begin{proof}
    We prove hardness by reducing from the Quantified Boolean Formula problem for $\Pi^p_2$, which consists in deciding whether $\forall X \exists Y \varphi(X,Y)$ is true or not, where $\varphi$ is a \CNF formula on the variables $X$ and $Y$. Let $X = \{x_1,\ldots,x_r\}$, $Y = \{y_1,\ldots,y_s\}$, and $c_1,\ldots,c_m$ be the clauses of $\varphi$. Moreover, for $z \in \{x,y\}$, if $z_i$ (resp.\ $\lnot z_i$) is the $k$-th literal of $c_j$, we write $l_j^k = \top^z_i$ (resp.\ $\bot_i^z$) to denote the assignment of that literal that makes the clause true.
    The intuition behind the construction of the data-graph is the following: we will define a graph $G_{\varphi}$ that encodes $\varphi(X,Y)$, and we will enforce through the constraints $R$ that every repair of $G_{\varphi}$ induces an assignment of the variables from $X$, and either (1) an assignment of the variables from $Y$ extending the one for the variables from $X$ that satisfies $\varphi(X,Y)$, if it exists, or (2) an empty assignment of the variables from $Y$. The query $q$ will detect if there exists some repair such that the variables from $Y$ are not assigned. In that case, we obtain that there is some assignment of the variables in $X$ such that no assignment for $Y$ makes the formula $\varphi(X,Y)$ true.\looseness=-1

    Let $V = V_{boolean}^x \cup V_{boolean}^y \cup V_{clause} \cup \{v,w\}$ and $E = E_{clause} \cup E_{one\_value} \cup E_{valid} \cup E_{all} \cup E_{query}$ be the nodes and edges of $G_\varphi$, respectively, where:
     \begin{align*}  
         V_{boolean}^x &= \{\star_i^x : 1 \leq i \leq r, \star \in \{\top, \bot\}\}\\
         V_{boolean}^y &= \{\star_i^y : 1 \leq i \leq s, \star \in \{\top,\bot\}\}\\
         V_{clause} &= \{c_j : 1 \leq j \leq m\}\\
         E_{clause} &= \{(c_j, \esLabel{needs}, l_j^k) : 1 \leq j \leq m, 1 \leq k \leq 3)\} \\
         &\;\;\; \cup \{(a, \esLabel{needs}, a) : a \in V \setminus V_{clause}\} \\
         E_{one\_value} &= \{(a, \esLabel{one\_v}, b) : a,b \in V\} \\
         &\;\;\; \setminus \Big(\{(\top_i^x, \esLabel{one\_v}, \bot_{i}^x) : 1 \leq i \leq r)\} \\
         &\;\;\;\;\;\;\;\; \cup \{(\top_i^y, \esLabel{one\_v}, \bot_{i}^y) : 1 \leq i \leq s)\}\Big)\\  
         E_{valid} &= \{(c_j, \esLabel{valid}, c_{j+1}) : 1 \leq j < m\}\\
         &\;\;\; \cup \{(c_m, \esLabel{valid}, \star_1^y) : \star \in \{\top, \bot\}\} \\
         &\;\;\; \cup \{(\star_i^y, \esLabel{valid}, \square_{i+1}^y) : 1 \leq i < s; \star,\square \in \{\top,\bot\}\} \\
         &\;\;\;\cup \{(\square_s^y, \esLabel{valid}, c_1) : \square \in \{\top, \bot\}\}\\
         &\;\;\; \cup \{(a, \esLabel{valid}, a) : a \in V \setminus \big(V_{boolean}^y \cup V_{clause}\big)\}\\      
         E_{all} &= \{(a, \esLabel{all}, b) : a,b \in V\}\\
         E_{query} &= \{(v, \esLabel{query}, c_1), (c_1, \esLabel{query}, w)\}
     \end{align*}
     \normalsize Fix the query as $q = \down_\esLabel{query} \down_\esLabel{query}$, and the constraints $R = \{\aPath_{satisfy}, \aPath_{one\_value}, \aPath_{valid}\}$ as:
     \vspace{-.2cm}
     \begin{align*}
         \aPath_{satisfy} &=\; \down_\esLabel{needs} \down_{\esLabel{all}}
         &\aPath_{one\_value} &=\; \down_\esLabel{one\_v}\\
         \aPath_{valid} &=\; \down_\esLabel{valid} \down_\esLabel{all}
     \end{align*}

    It can be seen, by the maximality of the subset repairs, that for any $G' \in \prefRepairsSet{\subseteq}{G_\varphi}{R}$ it holds that $\semantics{\down_\esLabel{all}}_{G'} = V_{G'} \times V_{G'}$ and for any $v \in V_{G'}$ it is the case that all self-loops that $v$ had in $G_{\varphi}$ have to remain in ${G'}$.

    Expression $\aPath_{one\_value}$ ensures that for each pair of nodes $(\top_i^z, \bot_i^z)$ with $z \in \{x, y\}$ at most one of them remains in any repair: note that in the definition of $E_{one\_value}$, these pairs do not belong to $\semantics{\aPath_{one\_value}}_{G_{\varphi}}$, and hence they will not belong to $\semantics{\aPath_{one\_value}}_{H}$ for any $H \subseteq G_{\varphi}$ due to the monotonicity of the expression. Thus, any $G' \in \prefRepairsSet{\subseteq}{G_{\varphi}}{R}$ induces a \textit{partial assignment} $\nu$ of the variables $X \cup Y$ defined as
    \begin{align*}
        \nu_{G'}(\star_i) = \begin{cases}
            \top & \top_i^\star \in V_{G'}\\
            \bot & \bot_i^\star \in V_{G'}\\
            undefined & \text{otherwise}
        \end{cases}
    \end{align*}
    
    Moreover, by maximality and the choice of constraints it holds that either $\top_i^x$ or $\bot_i^x$ will belong to $V_{G'}$: all other constraints are trivially satisfied for these nodes since they have $\esLabel{needs}$ and \esLabel{valid} loops.
     
     Expression $\aPath_{satisfy}$ ensures that whenever a clause $c_j$ is present in the repair $G'$ (i.e. $c_j \in V_{G'}$) then it is satisfied by the (partial) assignment induced by $G'$. Observe that, due to the definition of $E_{clause}$, all pairs $(a,b)$ where $a \in V \setminus V_{clause}$ satisfy $\aPath_{satisfy}$ trivially due to the $\esLabel{needs}$ loop at $a$. The remaining pairs $(c_j, b)$ will satisfy the expression if and only if at least one of the nodes that denotes an assignment that satisfies clause $c_j$ remains in the repair.

     Expression $\aPath_{valid}$ ensures that, either all variables from $Y$ have an assignment and all nodes $V_{clause}$ are in the repair, or rather none of the nodes $V_{clause}$ is in the repair and none of the variables from $Y$ has an assignment. 
     Notice that, due to $E_{valid}$, every pair of nodes $(a,b)$ such that $a \notin (V_{boolean}^y \cup V_{clause})$ satisfies this expression thanks to the \esLabel{valid} loop in $a$. 
     The pairs $(a,b)$ with $a \in V_{boolean}^y \cup V_{clause}$ will satisfy $\aPath_{valid}$ if and only if they keep their outgoing $\esLabel{valid}$ edge. This holds for all of them if and only if all clause nodes remain and at least one node of each $\{\top_i^y, \bot_i^y\}$ for every $1\leq i \leq s$.\looseness=-1

     Now we prove the correctness of the reduction. We show that $\forall X \exists Y \varphi(X,Y)$ is true if and only if $(v,w) \in \semantics{q}_{G'}$ for all $G' \in \repairs{\subseteq}{G_{\varphi}}{R}$.

     $\impliedby$) Assume that $(v,w) \in \semantics{q}_{G'}$ for all $G' \in \repairs{\subseteq}{G_\varphi}{R}$. To prove that $\forall X \exists Y \varphi(X,Y)$ is true, we will take an arbitrary assignment of $X$, define a data-graph $H \subseteq G_{\varphi}$ encoding that assignment, show that $H \models R$, pick a repair $G'$ with $H \subseteq G' \subseteq G_{\varphi}$, and finally show that $G'$ encodes an assignment of $Y$ such that $\phi(X,Y)$ holds. To prove this last fact, we will use the hypothesis that $(v,w) \in \semantics{q}_{G'}$.

     Let $\nu: X \to \{\top, \bot\}$ be a valuation of the variables from $X$. Then, define the data-graph $H = (V_H,E_H)$ as: 
    \vspace{-.15cm}
     \begin{align*}
         V_H &= \{\nu(x_i)_i^x : 1 \leq i \leq r\} \cup \{v,w\}\\
         E_H &= \{(a,\esLabel{l}, b): a,b \in V_H, \esLabel{l} \in \{\esLabel{one\_v}, \esLabel{all}\}\} \\
         &\;\;\; \cup \{(a,\esLabel{l}, a) : a \in V_H, \esLabel{l} \in \{\esLabel{needs}, \esLabel{valid}\}\}
     \end{align*}
    Since $\semantics{\down_{\esLabel{l}}}_H = V_H \times V_H$ for $l \in \{\esLabel{one\_v}, \esLabel{all}\}$, and $\semantics{\down_{\esLabel{l}}}_{H} = \{(a,a) : a \in V_H\}$ for $\esLabel{l} \in \{\esLabel{needs}, \esLabel{valid}\}$, it follows that $H \models R$. Also, note that $H \subseteq G_\varphi$, and, as a consequence, there is a data-graph $G' \in \repairs{\subseteq}{G_{\varphi}}{R}$ such that $H \subseteq G'$.
    Since we assumed $(v,w) \in \semantics{q}_{G'}$, it holds that $c_1 \in V_{G'}$. Furthermore, since $\aPath_{valid}$ is satisfied, then $V_{clause} \subseteq V_{G'}$, and moreover, at least one node of each pair $(\top_i^y,\bot_i^y)$ lies in $V_{G'}$. Due to the constraint $\aPath_{one\_value}$, at most one of each pair remains in a repair, and thus $G'$ induces a valuation over $X \cup Y$. Finally, since $\aPath_{satisfy}$ holds, we conclude that this assignment satisfies every clause.

    $\implies$) Assume that $\forall X \exists Y \varphi(X,y)$ is true. We will show that for every $G' \in \prefRepairsSet{\subseteq}{G_\varphi}{R}$ it holds that $c_1 \in V_{G'}$, which implies  that $(v,w) \in \semantics{q}_{G'}$ by the maximality of the repairs. Towards a contradiction, let us suppose that $c_1 \notin V_{G'}$ for some $G' \in \prefRepairsSet{\subseteq}{G_\varphi}{R}$, which implies that $(V_{clause} \cup V_{boolean}^y) \cap V_{G'} = \emptyset$ by $\aPath_{valid}$.

    Observe that, for each $1 \leq i \leq r$, it follows from the maximality of the repairs that either $\top_i^x \in V_{G'}$ or $\bot_i^x \in V_{G'}$: otherwise, we could add any of them together with its loops and its edges toward the rest of the nodes of $G'$ to obtain a new data-graph $G''$ such that $G' \subset G''$ and $G'' \models R$. Thus, $G'$ induces the valuation $\nu: X \to \{\top, \bot\}$ defined as
    \begin{align*}
        \nu(\star_i) = \begin{cases}
            \top & \top_i^x \in V_{G'}\\
            \bot & \bot_i^x \in V_{G'}
        \end{cases}
    \end{align*}

    Since $\forall X \exists Y \varphi(X,Y)$ is true, there is a valuation $\nu' : X\cup Y \to \{\top, \bot\}$ such that $\nu'|_{X} = \nu$ and $\nu'$ satisfies $\varphi(X,Y)$. Then, we can consider the data-graph $G''$ induced by the nodes $V_{G'} \cup V_{clause} \cup \{\nu'(y_i)_i^y : 1 \leq i \leq s\}$. It is straightforward to check that $G' \subset G''$ and $G'' \models R$, which implies that $G'$ is not a repair, thus reaching a contradiction.
    This concludes the proof for the first case. 
    
    Finally, we show a general strategy to transform a set of $\Gposregxpath$ path constraints to a set of $\GNodeXPath$ expressions, which can be used to define the set $R'$ from the Theorem's statement. 
    Given $\alpha \in \Gposregxpath$, we define $\alpha_T \in \GNodeXPath$ as $\aPath_T = \lnot \comparacionCaminos{\pathComplement{\aPath}}$. Observe that given any data-graph $H$, $H \models \aPath \iff H \models \aPath_T$: if $(v,w) \notin \semantics{\aPath}_H$ then $v \in \semantics{\comparacionCaminos{\pathComplement{\aPath}}}_H$, and consequently $v \notin \semantics{\lnot \comparacionCaminos{\pathComplement{\aPath}}}_H$. Similarly, if $v \notin \semantics{\lnot \comparacionCaminos{\pathComplement{\aPath}}}_H$ then $(v,w) \in \semantics{\pathComplement{\aPath}}_H$ for some $w$, and thus $(v,w) \notin \semantics{\aPath}_{H}$. Therefore, we may consider $R' = \{\aPath_T : \aPath \in R\}$ (where the set $R$ is the set defined in the last construction) to obtain the desired hardness result.\looseness=-1
\end{proof}
%\vspace{-1cm}
The previous proof can be adapted for the multiset criterion. \footnote{Details can be found in the supplementary material.}. 

\begin{theorem}\label{teo:multiset-criteria-hard}
    There exists a set $R \subseteq \Gposregxpath$, a query $q$, and a partial order $(\Sigma_e \cup \Sigma_n, \leq)$, such that the problem \problemCQAsubset{\multisetCriteria} is \piptwo{}\textit{-hard}. Moreover, there exists a set of constraints $R' \subseteq \GNodeXPath$ such that \problemCQAsubset{\multisetCriteria} is \piptwo{}\textit{-hard}.
\end{theorem}

\textit{Sketch of the proof}: We use the same construction from the previous proof, adding data values to the nodes and defining a trivial partial order such that $\prefRepairsSet{\subseteq}{G_\varphi}{R} = \prefRepairsSet{\multisetCriteria}{G_\varphi}{R}$, which suffices to obtain the result.\hfill \qedsymbol{}

\vspace{0.2cm}

Now we prove the \deltaptwo $[\log n]$\textit{-hardness} for the problem \problemCQAsubset{\cardinalityCritera} for fixed sets $R \subseteq \Gposregxpath$ and $R' \subseteq \GNodeXPath$. It is based on a reduction from the \deltaptwo$[\log n]$ problem \textsc{PARITY(3SAT)} \cite{eiter1997complexity}.

\begin{theorem}\label{teo:deltaptwo_hardness_CQA_cardinality}
    There is a set $R \subseteq \Gposregxpath$ and a query $q$ such that the problem $\problemCQAsubset{\cardinalityCritera}$ is \deltaptwo$[\log n]$\textit{-hard}. There is also a set of constraints $R' \subseteq \GNodeXPath$ that achieves the same hardness result.   
\end{theorem}

\textit{Sketch of the proof}: The problem \textsc{PARITY(3SAT)} takes as an input a set of \CNF{} formulas $\{\varphi_1,\ldots,\varphi_k\}$ and the task consists in deciding whether the number of satisfiable formulas is even. It can be assumed without loss of generality that $\varphi_{\ell+1}$ is unsatisfiable whenever $\varphi_{\ell}$ is unsatisfiable~\cite{wagner1987more}.~\looseness=-1

Our reduction defines a data-graph $G$ encoding a representation of each $\varphi_\ell$ similarly as in Theorem~\ref{teo:pi_p_2_hardness}. The constraints $R$ will enforce that every repair encodes a satisfying assignment for each $\varphi_\ell$. Moreover, whenever $\varphi_\ell$ cannot be satisfied a node $f_\ell$ will be in the repair indicating that failure. The query $q$ will be used to detect for which $\ell$ it holds that $\varphi_\ell$ is satisfied and $f_{\ell + 1}$ is in the repair, indicating that the answer to the problem is the parity of $\ell$. The preference criteria will be useful to ensure that $f_\ell$ is present if and only if $\varphi_\ell$ is unsatisfiable. \hfill \qedsymbol{}

\vspace{0.2cm}

The \deltaptwo-hardness proof for \problemCQAsubset{\weightCriteria} and \problemCQAsubset{\priorityzedCardinalityCriteria} is obtained via a reduction from the following \deltaptwo-complete problem \cite{krentel1986complexity}: given a \CNF formula $\varphi$ on clauses $c_1,\ldots,c_m$ and variables $x_1,\ldots,x_n$ decide whether the lexicographically maximum truth assignment $\nu$ satisfying $\varphi$ fulfills the condition $\nu(x_n) = \top$ (i.e. whether the last variable is assigned \texttt{true}).\looseness=-1

\begin{theorem}\label{teo:deltaptwohardness_weight_prioritized_cardinality}

    There is a set $R \subseteq \Gposregxpath$ and a query $q$ such that the problems \problemCQAsubset{\weightCriteria} and \problemCQAsubset{\priorityzedCardinalityCriteria} are \deltaptwo-hard. There is also a set of constraints $R' \subseteq \GNodeXPath$ that achieves the same hardness result.
    
\end{theorem}

\textit{Sketch of the proof}: Given a \CNF{} formula $\varphi$ we will define a data-graph $G_{\varphi}$ in the same way as in Theorem~\ref{teo:pi_p_2_hardness}. Thus, by choosing the suitable constraints we will ensure that every repair represents a satisfying assignment of $\varphi$, if there is one. Using the preference criteria we will guarantee that data-graphs representing a lexicographically larger assignment are preferable. \hfill \qedsymbol{}

\section{Conclusions and Future Work}\label{sec:conclusions}

%In this work, over a framework of graph database models and a notion of consistency based on \Gregxpath constrains, we defined the problems associated with database repairing and CQA, with prioritizations given trhough preorders. 

In this work, we studied prioritized repairs for a data-graphs with integrity constraints based on the \Gregxpath language. 
We defined the problems associated with database repairing and CQA for this data model and a prioritization given through a preorder, 
%Given a set of restrictions, 
and performed a systematic study of the complexity of these problems for prioritized (subset) repairs of data-graphs, considering some of the most prominent preorders studied in the literature for other data models, and multisets (results are summarized in Tables~\ref{tab:repair_results} and~\ref{tab:cqa_results}). 
We showed that, depending on the preference criteria and the syntactic conditions imposed to the set of constraints, the repair problem ranges from polynomial-time solvable to \NP-complete/\coNP-complete, while CQA ranges from \PTIME to $\Pi_2^p$-complete.
When restricting the language to the positive fragment of \Gregxpath we obtained polynomial-time algorithms (in data complexity) for prioritized subset repairs considering only node expressions. 

Some questions in this context remain open, such as that of finding refined tractable versions of the prioritized subset repair problem that might be based on real-world applications. On the other hand, the notion of consistency that we use could be influencing negatively the complexity of the problems. An alternative  is to consider local notions of consistency that do not require the evaluation of  queries or constraints on the entire data-graph, such as those studied for Description Logics~\cite{Wasserman99-DL}, and build up from the positive fragment of \Gregxpath as a starting point to find tractable versions. 
Alternative definitions of repair semantics with preferences for data-graphs~\cite{barcelo2017data}, based on the well-known superset and symmetric-difference repairs~\cite{Cate:2015} are also worth studying in the preference-based setting. The same systematic study remains open for the complexity of the \textsc{prioritized superset repair} problem with constraints expressed in \Gregxpath; it does not seem straightforward to extend the proofs and techniques used for subset repairs in the case of path expressions or node expressions (under unfixed constraints).
It would also be interesting to study more general families of criteria, such as the ones proposed in~\cite{staworko2012prioritized}, and analyze whether the complexity of the decision problems changes for data-graphs. Another avenue of future research is the development of more fine-grained preference criteria over data-graphs, where the topology of the graph is exploited and taken into consideration in the prioritization by introducing preferences through particular shapes or patterns, or even graph parameters such as bounded treewidth or pathwidth.
In the definition of preferred repairs presented in this work, we never discuss how to handle scenarios with multiple solutions. While introducing a notion of preference reduces the set of repairs of interest, this does not address the problem of choosing one in the presence of several options. In practice, the preference criterion should be adjusted based on the general use case, in order to reduce the set of obtained repairs to a size that is acceptable in expectation. 
The complexity of CQA for prioritized superset and symmetric-difference repairs in this context has not been studied yet. Taking into account the results from~\cite{Cate:2015} and~\cite{barcelo2017data}, it is plausible that reasoning over superset repairs will be harder than in subset repairs. However, in the meantime, the complexity of the problem remains unknown. 

 %In the problems we studied, we were concerned in finding \textit{any} repair of a data-graph given a set of constraints. This assumes no prior information about the domain semantics. However, in many contexts we might have a preference criterion that may yield an ordering over all possible repairs.

%Changing the consistency semantics...
%Origin semantics... Nina: No puse nada de esto
%\nina{acortar un poco esto y ?sumar lo anterior?}

\bibliographystyle{kr}
\bibliography{biblio}

%\nina{OJO que el appendix que se ponga en la submission de KR tiene que estar dentro de las 9 pags, sino hay que sumarlo como pdf separado a la submission pero no va a ir a la version final del paper.}

\clearpage
\newpage

\appendix

\section{Appendix}

\textit{Proof of Theorem~\ref{teo:multiset-criteria-hard}}: Consider the construction from Theorem~\ref{teo:pi_p_2_hardness}. We will specify data values for the nodes of $G_{\varphi}$ and a partial order such that $\prefRepairsSet{\multisetCriteria}{G_\varphi}{R} = \prefRepairsSet{\subsetInclusionCriteria}{G_\varphi}{R}$, which suffices to obtain the result.
Let $\{\esDato{null}\} \cup \{\star_i^z : \star \in \{\top, \bot\}, i \in \N, z \in \{x,y\}\} \subseteq \Sigma_n$ be data values from $\Sigma_n$. We define the data function $D$ as:
\vspace{-.15cm}
\begin{align*}
    D(\star_i^x) =& \star_i^x \;\;\; \text{ for } \star \in \{\top, \bot\}, 1 \leq i \leq r\\
    D(\star_i^y) =& \star_i^y \;\;\; \text{ for } \star \in \{\top, \bot\}, 1 \leq i \leq s\\ 
    D(u) =& \esDato{null}  \text{ for } u \in V \setminus \left( V_{boolean}^x \cup V_{boolean}^y \right)
\end{align*}
Consider the trivial partial order over $\Sigma_e \cup \Sigma_n$ defined as $\leq \; := \{(a, a) : a \in \Sigma_e \cup \Sigma_n\}$. 
It follows from the multiset semantics that $G_1 \multisetCriteria G_2$ if $G_1^\mathcal{M}(x) \leq G_2^\mathcal{M}(x)$ for all $x \in \Sigma_e \cup \Sigma_n$.

We know that $\prefRepairsSet{\multisetCriteria}{G_\varphi}{R} \subseteq \prefRepairsSet{\subsetInclusionCriteria}{G_\varphi}{R}$ due to Lemma~\ref{obs:all_pref_rep_are_subset}. To prove the other direction, let $G',G'' \in \prefRepairsSet{\subseteq}{G_{\varphi}}{R}$, with $G' \neq G''$. We will show that $\lnot (G'' \multisetCriteria G')$, which implies that $G' \in \prefRepairsSet{\multisetCriteria}{G_{\varphi}}{R}$.

If $V_{boolean}^x \cap V_{G'} \neq V_{boolean}^x \cap V_{G''}$ then $G'$ and $G''$ are incomparable according to $\multisetCriteria$ since the nodes in $V_{boolean}^x$ all have different data values. Otherwise, if $V_{G'} \cap V_{boolean}^x = V_{G''} \cap V_{boolean}^x$ there are two cases:

\begin{itemize}
    \item If the valuation for the variables $X$ that they induce cannot be extended to a satisfying assignment for $\varphi(X,Y)$ then $c_1 \notin V_{G'}, V_{G''}$ and it holds that $V_{G'} = V_{G''}$. By maximality of the repairs this implies that $G' = G''$, which is a contradiction.

    \item In the other case it holds that $c_1 \in V_{G'}, V_{G''}$, and thus both $G'$ and $G''$ induce a valuation for the variables of $Y$ such that $\varphi(X,Y)$ is true. As noted before, if they induce different valuations then they are incomparable with respect to $\multisetCriteria$, while if they induce the same one then $G' = G''$. \hfill \qedsymbol{}

\end{itemize}
\vspace{0.2cm}

\textit{Proof of Theorem~\ref{teo:deltaptwo_hardness_CQA_cardinality}}: We prove this theorem through a reduction from the \deltaptwo$[\log n]$ problem \textsc{PARITY(3SAT)} \cite{eiter1997complexity}. Its input is a set $\{\varphi_1,\ldots,\varphi_k\}$ of \CNF formula (each one on different variables) and the problem consists on deciding whether the number of satisfiable formulas is even. Moreover, it can be assumed without loss of generality that whenever $\varphi_\ell$ is unsatisfiable then $\varphi_{\ell + 1}$ is unsatisfiable as well \cite{wagner1987more}. Thus, the problem reduces to finding the value $1 \leq t \leq k$ such that $\varphi_t$ is satisfiable and $\varphi_{t+1}$ is not. We will assume without loss of generality that each \CNF formula $\varphi_\ell$ has at least two clauses.

Our reduction is based on the tools and ideas developed in the previous proofs: given an input $\{\varphi_1,\ldots,\varphi_k\}$ of \textsc{PARITY(3SAT)}, where each formula $\varphi_\ell$ has clauses $c_1^\ell,\ldots, c_{m_\ell}^\ell$ and variables $x_1^\ell,\ldots,x_{n_\ell}^\ell$, we will build a data-graph $G = (V,L,D)$ containing a representation of all formulas $\varphi_j$ along its variables, two special nodes $v,w$ used for the query and some extra nodes $f_\ell$. Through the constraints we will enforce that, for each $\varphi_\ell$, every repair either (1) represents a satisfying valuation of it, or (2) represents a non-satisfying assignment and deletes the clause nodes. Because of the $\cardinalityCritera$ preference criteria the repairs with a satisfying assignment will be preferred (they have more nodes), and through the query $q$ we will detect how many of the formulas $\varphi_\ell$ are satisfied in these repairs.

More formally, let $V = V_{clause} \cup V_{var} \cup V_{finish} \cup \{v,w\}$ and $E = E_{clause} \cup E_{one\_value} \cup E_{valid} \cup E_{all} \cup E_{finish} \cup E_{query}$ where: 
\begin{align*}
%\footnotesize
    V_{clause} &= \bigcup_{\ell=1}^k V_{clause}^\ell\\
    V_{clause}^\ell &= \{ c_j^\ell : 1 \leq j \leq m_\ell \}\\
    V_{var} &= \bigcup_{\ell=1}^k V_{boolean}^\ell\\
    V_{boolean}^\ell &= \{\star_i^\ell : 1 \leq i \leq n_\ell, \star \in \{\bot, \top\}\}\\
    V_{finish} &= \{f_\ell : 1 \leq \ell \leq k\}\\
    E_{clause} &= \Big( \bigcup_{\ell=1}^k  E_{clause}^\ell \Big) \cup \{(a, \esLabel{needs}, a) : a \in V \setminus V_{clause}\}\\
    E_{clause}^\ell &= \{(c_j^\ell, \esLabel{needs}, l_{j,p}^\ell) : 1 \leq j \leq m_j, 1 \leq p \leq 3 \text{ where } \\
    &\;\;\;l_{j,p}^\ell = \top_i^\ell \text{ (resp. } \bot_i^\ell \text{)} \text{ if } x_i^\ell \text{ is the } pth \text{ literal of } c_j^\ell \\ &\;\;\; \text{(resp. } \lnot x_i^\ell\text{)}\}\\
    E_{one\_value} &= \{(a, \esLabel{one\_v}, b) : a, b \in V\}\\
    &\;\;\;\setminus \left( \bigcup_{\ell=1}^k \{(\top_i^\ell, \esLabel{one\_v}, \bot_i^\ell) : 1 \leq i \leq n_\ell\} \right)\\
    E_{valid} &= \left(\bigcup_{\ell=1}^k E_{valid}^\ell\right) \cup \{(a, \esLabel{valid}, a) : a \in V \setminus V_{clause}\}\\
    E_{valid}^\ell &= \{(c_j^\ell, \esLabel{valid}, c_{j+1}^\ell) : 1 \leq j < m_\ell\} \\
    &\;\;\; \cup \{(c_{m_\ell}^\ell, \esLabel{valid}, c_1^\ell)\} \\
    E_{all} &= \{(a, \esLabel{all}, a) : a \in V)\}\\
    E_{finish} &= \{(a, \esLabel{finish}, b) : a, b \in V\} \\
    &\;\;\; \setminus \left( \{(f_\ell, \esLabel{finish}, c_1^{\ell+1}) : 1 \leq \ell < k\} \right)\\
    E_{query} &= \{(v, \esLabel{query}, c_1^\ell) : 1 \leq \ell \leq k, \ell = 0 \bmod 2\}\\
    &\;\;\;\; \cup \{(c_1^\ell, \esLabel{query}, f_\ell) : 1 \leq \ell \leq k, \ell = 1 \bmod 2\} \\
    &\;\;\;\; \cup \{(f_\ell, \esLabel{query}, w) : 0 \leq \ell \leq k\} %\\
    %&\;\;\;\;
\end{align*}
\normalsize Fix the query as $q = \down_{\esLabel{query}} \down_{\esLabel{query}} \down_{\esLabel{query}}$ and the constraints $R = \{\aPath_{satisfy}, \aPath_{one\_value}, \aPath_{valid}, \aPath_{finish}\}$ as:
\begin{align*}
    \aPath_{satisfy} &= \down_{\esLabel{needs}} \down_{\esLabel{all}} &\aPath_{one\_value} &= \down_\esLabel{one\_v}  \\
    \aPath_{valid} &= \down_{\esLabel{valid}} \down_{\esLabel{all}} &\aPath_{finish} &= \down_{\esLabel{finish}}
\end{align*}

The first three constraints work in the same way as the ones from Theorem~\ref{teo:pi_p_2_hardness}. Meanwhile, the last one enforces that both $f_\ell$ and $c_1^{\ell+1}$ cannot remain in the repair. Observe that if $c_1^\ell$ is deleted then all other nodes $c_j^\ell$ must be deleted as well because of the constraint $\aPath_{valid}$.
The following lemma justifies the constructions:

\begin{lemma}\label{lemma:construction_parity_3sat}
    $\varphi_\ell$ is satisfiable if and only if the node $c_1^\ell$ lies in all $\cardinalityCritera$-preferred subset repairs of $G$ with respect to $\aRestrictionSet$.
\end{lemma}

\begin{proof}
    $\implies)$ Suppose $\varphi_\ell$ is satisfied by the valuation $\nu:\{x_i^\ell : 1 \leq i \leq n_\ell\} \to \{\top, \bot\}$ and that there is some $\cardinalityCritera$-preferred subset repair $G'$ such that $c_1^\ell$ does not belong to the repair. Thus $c_j^\ell \notin V_{G'}$ for all $1 \leq j \leq m_\ell$ due to the constraint $\aPath_{valid}$, and the assignment encoded by the nodes $V' = V_{boolean}^\ell \cap V_{G'}$ is a non-satisfying one for $\varphi_\ell$. But then we can perform the following swap to obtain a subset repair which is preferable according to the $\cardinalityCritera$ preference criteria: remove the nodes $V'$ and the node $f_\ell$ from $G'$\footnote{We know that $f_\ell \in V_{G'}$ because $G'$ is maximally consistent.} and add instead the nodes $\{\nu(x_i^\ell)_i^\ell : 1 \leq i \leq n_\ell\} \cup V_{clause}^\ell$ along all its outgoing edges from $G$ to nodes that are still in $V_{G'} \setminus V'$. It can be checked that the obtained data-graph satisfies all constraints from $R$, and to conclude that it is preferred note that the number of nodes increased (because $m_\ell > 1)$, the outdegree of all added nodes from $V_{boolean}^\ell$ is bigger than or equal to the outdegree of the nodes from $V'$, and the outdegree of each $c_j^\ell$ is bigger than or equal to the one from $f_\ell$.

    $\impliedby)$ If $\varphi_\ell$ is unsatisfiable then $c_1^\ell$ does not belong to any repair of $G$ with respect to $\aRestrictionSet$ because of the imposed constraints. The argument is analogous to the one from Theorem~\ref{teo:pi_p_2_hardness}.
\end{proof}

Now we argue the correctness of the construction: we claim that the unique $t$ such that $\varphi_t$ is satisfiable and $\varphi_{t + 1}$ is unsatisfiable is even if and only if $(v,w) \in \semantics{q}_{G'}$ for all $G' \in \prefRepairsSet{\cardinalityCritera}{G}{R}$. This is a direct consequence of Lemma~\ref{lemma:construction_parity_3sat}: $\varphi_t$ is satisfiable and $\varphi_{t+1}$ is unsatisfiable if and only if $c_1^t$ belongs to all $\cardinalityCritera$-preferred subset repairs of $G$ with respect to $R$, while $c_1^{t+1}$ does not belong to a single one. By the constraint $\aPath_{finish}$ this implies that $f_t$ belongs to all repairs, and as a consequence if $t$ is even then the path $v \down_{\esLabel{query}} c_1^t \down_{\esLabel{query}} f_t \down_{\esLabel{query}} w$ belongs to all repairs. Meanwhile, if $t$ is odd then there is no path that satisfies the query, since for all $\ell < t$ it is the case that the node $f_\ell$ is not in the data-graph, while for $\ell > t$ the node $c_1^\ell$ is the one missing.  \hfill \qedsymbol{}

\vspace{0.2cm}  

\textit{Proof of Theorem~\ref{teo:deltaptwohardness_weight_prioritized_cardinality}}: We consider the case for $\problemCQAsubset{\weightCriteria}$, and explain how to adapt the proof for the case of the $\priorityzedCardinalityCriteria$ preference criteria. The construction is similar to that in the previous proof. Given a \CNF formula $\varphi$ on clauses $c_1,\ldots,c_m$ and variables $x_1,\ldots,x_n$ we define a data-graph $G = (V, L, D)$ containing a representation of $\varphi$. The constraints will enforce that every repair represents a satisfying assignment of $\varphi$ (if such assignment exists), and via the prioritization we will enforce this repair to contain a lexicographically maximum assignment. Finally, using the query $q$ we will detect if the node $\top_{n}$ remained in the repair.

Let $\{i : i \in \N\} \subseteq \Sigma_n$. We define $V = V_{clause} \cup V_{var} \cup \{v,w\}$ and $E = E_{clause} \cup E_{one\_value} \cup E_{valid} \cup E_{all} \cup E_{query}$ as
\begin{align*}
    %\footnotesize
    V_{var} &= \{\star_i : 1 \leq i \leq n, \star \in \{\top, \bot\} \}\\
    V_{clause} &= \{c_j : 1 \leq j \leq m\}\\
    E_{clause} &= \{(c_j, \esLabel{needs}, l_j^k) : 1 \leq j \leq m, 1 \leq k \leq 3 \text{ where }\\
    &\;\;\;\; l_j^k = \top_i \text{ ( resp. }  \bot_i \text{) if } x_i \text{ (resp. } \lnot x_i \text{) is the } k \text{th}\\
    &\;\;\;\; \text{ literal of } c_j\}\\
    E_{one\_value} &= \{(a, \esLabel{one\_v}, b) : a, b \in V\} \\
    &\;\;\;\; \setminus \{(\top_i, \esLabel{one\_v}, \bot_i) : 1 \leq i \leq n\}\\
    E_{valid} &= \{(c_j, \esLabel{valid}, c_{j+1}) : 1 \leq j < m\} \\
    &\;\;\;\; \cup \{(c_m, \esLabel{valid}, c_1)\}\\
    E_{all} &= \{(a, \esLabel{all}, b) : a,b \in V\}\\
    E_{query} &= \{(v, \esLabel{query}, \top_n), (\top_n, \esLabel{query}, w)\}
\end{align*}

\normalsize The data values are assigned as
\begin{align*}
    D(\top_i) &= 2^{n+1-i}\\
    D(\bot_i) &= 1\\
    D(c_j) &= 2^{n+2}
\end{align*}

and the weights for the data values as $w(i) = i$, while $w(\aLabel) = 0$ for all $\aLabel \in \Sigma_e$.
The query is fixed as $q = \down_\esLabel{query} \down_\esLabel{query}$ and the constraints $R = \{\aPath_{satisfy} \aPath_{one\_value}, \aPath_{valid}\}$ identically as in Theorem~\ref{teo:pi_p_2_hardness}.

As before, due to $\aPath_{one\_value}$ it holds that every repair induces an assignment of the variables $x_1,\ldots,x_n$. Since the clause nodes $c_1,\ldots,c_m$ have weight $2^{n+2} > \sum_{i=1}^n 2^{n+1-i}$ repairs representing satisfying assignments are always preferred to non-satisfying ones. Also, observe that the weights chosen for the nodes $\top_1,\bot_1,\ldots,\top_n,\bot_n$ ensure that from the set of satisfying assignments the one preferred will be the lexicographically maximum one. Thus, it can be proven that there is only one subset repair: $\emptyset$ if $\varphi$ is unsatisfiable and the data-graph $G' \subseteq$ representing the lexicographically maximum satisfying assignment otherwise. In the latter case, if $G'$ is the unique subset repair it holds that $G', v, w \models q$ if and only if $\top_n \in V_{G'}$ which implies that $\nu(x_n) = \top$ in the lexicographically maximum satisfying assignment.

When considering the $\priorityzedCardinalityCriteria$ criteria we can enforce an analogous preference by defining the prioritization
\begin{align*}
    &P_1 = \{c_1,\ldots,c_m,v,w\} \cup E\\
    &P_{i+1} = \{\top_i, \bot_i\} \text{ for } 1\leq i \leq n\\
\end{align*}

\end{document}

%% file: macros.tex
\usepackage{amsthm}
\theoremstyle{plain}

\newcommand{\eqdef}{\mathrel{{\mathop:}}=}

\newcommand{\defstyle}[1]{{\emph{#1}}}

\newcommand{\set}[1]{\{#1\}}
\newcommand{\tup}[1]{\langle#1\rangle}

\newcommand{\aGraph}{G}%
\newcommand{\aNode}{x}%
\newcommand{\aNodeb}{y}%
\newcommand{\aLabel}{\esLabel{a}} %Para edge labels

\newcommand{\aData}{\esDato{c}} %un data label /data value
\newcommand{\aFormula}{\ensuremath{\varphi}\xspace}
\newcommand{\aNodeExpression}{\aFormula}
\newcommand{\aFormulab}{\ensuremath{\psi}\xspace}
\newcommand{\aNodeExpressionb}{\aFormulab}

\newcommand{\aPath}{\ensuremath{\alpha}\xspace}
\newcommand{\aPathb}{\ensuremath{\beta}\xspace}

\newcommand{\aRestrictionSet}{\ensuremath{R}\xspace}%
\newcommand{\aRestrictionSetPaths}{\ensuremath{P}\xspace}%
\newcommand{\aRestrictionSetNodes}{\ensuremath{N}\xspace}%

\newcommand{\wrt}{\text{w.r.t.\ }}

\definecolor{darkred}{rgb}{0.55, 0.0, 0.0}
\usepackage[backgroundcolor=orange!20, textcolor={darkred}, textsize=tiny]{todonotes}
\definecolor{green}{RGB}{0,120,0}
\definecolor{blue}{RGB}{0,120,150}
\definecolor{hlyellow}{RGB}{250, 250, 190}
\definecolor{editcolor}{RGB}{100,190,165}
\newcommand{\NP}{\textsc{NP}\xspace}
\newcommand{\coNP}{\textsc{coNP}\xspace}
\newcommand{\coNPhard}{\textsc{coNP}-hard\xspace}
\newcommand{\coNPcomplete}{\textsc{coNP}-c\xspace}
\newcommand{\NPhard}{\textsc{NP}-hard\xspace}
\newcommand{\NPcomplete}{\textsc{NP}-c\xspace}

\newcommand{\PTIME}{\textsc{PTIME}\xspace}
\newcommand{\piptwo}{$\Pi^p_2$}
\newcommand{\piptwoComplete}{$\Pi^p_2$-c}

\newcommand{\deltaptwo}{$\Delta^p_2$}
\newcommand{\deltaptwoComplete}{$\Delta^p_2$-c}\newcommand{\deltaptwologComplete}{$\Delta^p_2[\log n]$-c}

\newcommand{\problemRepairExistence}[1]{$\exists_{#1}$\textsc{repair}}
\newcommand{\problemRepairChecking}[1]{$#1$-\textsc{repairChecking}\xspace}
\newcommand{\problemRepairComputing}[1]{$#1$-\textsc{repairComputing}\xspace}

\newcommand{\problemCQAsubset}[1]{\textsc{CQA$_{#1}$-repair}}

\newcommand{\Gregxpath}{\text{GXPath}\xspace}

\newcommand{\Gposregxpath}{\text{GXPath$^{pos}$}\xspace}

\newcommand{\GNodeXPath}{\text{GXPath$^{node}$}}
\newcommand{\GNodePosXPath}{\text{GXPath$^{node-pos}$}}

\newcommand{\down}{\downarrow\xspace}

\newcommand{\expNodoEnCamino}[1]{\ensuremath{[#1]}}
\newcommand{\comparacionCaminos}[1]{\ensuremath{\langle#1\rangle}}

\newcommand{\esLabel}[1]{{\textsc{#1}}} % para edge-labels
\newcommand{\esDato}[1]{{\textsc{\footnotesize{#1}}}} %Para representar un dato/data value/data label
\newcommand{\esDatoIgual}[1]{{\ensuremath{\esDato{#1}^=}}} %Operacion de nodo que dice que ahi vale ese dato
\newcommand{\esDatoDistinto}[1]{{\ensuremath{\esDato{#1}^{\neq}}}}

\newcommand{\dataFunction}{D} %El nombre de la función que usamos para asignar valores de datos

\newcommand{\complementoCamino}[1]{\ensuremath{\overline{#1}}}
\newcommand{\pathComplement}[1]{\complementoCamino{#1}}

\newcommand{\union}{\ensuremath{\cup}}

\newcommand{\unionCaminos}{\ensuremath{\cup}}
\newcommand{\pathUnion}{\unionCaminos}

\newcommand{\pathIntersection}{\ensuremath{\cap}}

\newcommand{\entoncesNodo}{\ensuremath{\Rightarrow}} %\rightarrow
\newcommand{\entoncesCamino}{\ensuremath{\Rrightarrow}} %\Rightarrow

\newcommand{\semantics}[1]{\ensuremath{\llbracket{#1}\rrbracket}} %

\newcommand{\N}{\mathbb{N}}

\newcommand{\parts}[1]{\mathcal{P}(#1)} %Conjunto de partes

\newcommand{\labelComodin}{\_} % representa bajar por cualquier Edge/Eje

\newcommand{\repairs}[3]{Rep_{#1}(#2,#3)}

\newcommand{\prefrepair}[1]{\ensuremath{{#1}\text{-repair}}}\xspace%

\newcommand{\weight}{\ensuremath{w}}

\newcommand{\multiset}[1]{\ensuremath{\mathcal{M}_{<\infty}(#1)}} %Multisets finitos

\newcommand{\lessEqualMultiset}{\ensuremath{\le_{mset}}}

\newcommand{\edgeDataMultiset}[1]{\ensuremath{#1^{\mathcal{M}}}}

\newcommand{\multisetCriteria}{\ensuremath{\preccurlyeq_{\mathcal{M}}}}

\newcommand{\queryLanguage}{\mathcal{Q}}
\newcommand{\prefRepairsSet}[3]{\text{Rep}_{#1}(#2, #3)}

\newcommand{\subsetInclusionCriteria}{\subseteq}
\newcommand{\cardinalityCritera}{\leq}
\newcommand{\priorityzedInclusionCriteria}{\subseteq_P}
\newcommand{\weightCriteria}{\leq_w}
\newcommand{\priorityzedCardinalityCriteria}{\leq_P}

\theoremstyle{definition}
\newtheorem{example}{Example}
\theoremstyle{plain} %Nota amsthm es necesario para Theoremstyle
\newtheorem{theorem}{Theorem}
\newtheorem{proposition}[theorem]{Proposition}
\newtheorem{lemma}[theorem]{Lemma}
\newtheorem{corollary}[theorem]{Corollary}
\newtheorem{observation}[theorem]{Observation}

\newtheorem{definition}[theorem]{Definition}
\newtheorem{remark}[theorem]{Remark}

\newcommand{\CNF}{\textsc{3CNF}\xspace}